\documentclass[a4paper]{article}
\usepackage{fullpage}

\newcommand{\instancename}[1]{\ensuremath{\mathsf{#1}}} 
\newcommand*{\LPF} {\instancename{LPF}}
\newcommand*{\LCP} {\instancename{LCP}}
\newcommand*{\PLCP}{\instancename{PLCP}}
\newcommand*{\ISA} {\instancename{ISA}}
\newcommand*{\SA}  {\instancename{SA}}
\newcommand*{\BWT}  {\instancename{BWT}}

\newcommand*{\strategyname}[1]{{{\renewcommand{\rmdefault}{phv}\fontfamily{phv}\selectfont\textrm{\textup{#1}}}}} 
\newcommand*{\functionname}[1]{{{\renewcommand{\rmdefault}{ptm}\fontfamily{ppl}\selectfont\textrm{\textup{#1}}}}} 
\newcommand*{\lcp}         [1][]{\UnaryOperator[#1]{\functionname{lcp}}}

\newcommand*{\ibeg}[1]{\mathsf{b}(#1)}
\newcommand*{\iend}[1]{\mathsf{e}(#1)}

\newcommand*{\threshold}{\ensuremath{\xi}}
\newcommand*{\dst}{\mathit{dst}}
\newcommand*{\src}{\mathit{src}}
\newcommand*{\Cand}{\mathcal{C}}

\newcommand*{\exampleString}{\texttt{ababbabababbabbaababa\$}}

\newcommand*{\strScanName}{scan}
\newcommand*{\strScan}{\strategyname{\strScanName}}
\newcommand*{\strPJName}{EM-PJ}
\newcommand*{\strPJ}{\strategyname{\strPJName}}
\newcommand*{\strIMCompactName}{IM-Compact}
\newcommand*{\strIMCompact}{\strategyname{\strIMCompactName}}
\newcommand*{\strEMCompactName}{EM-Compact}
\newcommand*{\strEMCompact}{\strategyname{\strEMCompactName}}

\newcommand*{\iPlcpcomp}{\strategyname{plcpcomp}}
\newcommand*{\iLcpcomp}{\strategyname{lcpcomp}}
\newcommand*{\iEMLPF}{\strategyname{EM-LPF}}
\newcommand*{\iEMBWT}{\strategyname{pEM-BWT}}
\newcommand*{\iLZSS}{\strategyname{LZ77}}
\newcommand*{\iLZend}{\strategyname{LZ-End}}
\newcommand*{\iLZscan}{\strategyname{EM-LZscan}}
\newcommand*{\iSEKKP}{\strategyname{SE-KKP}}

\usepackage{amsthm,amsmath,amssymb}
\usepackage{hyperref}
\usepackage{adjustbox}
\usepackage{graphicx}

\usepackage{microtype}
\usepackage{subcaption}
\usepackage{enumitem}
\usepackage{booktabs}
\usepackage{amssymb,amsmath,amsthm}

\RequirePackage{amsthm} %
\RequirePackage{thmtools} %

\theoremstyle{definition}
\newtheorem{theorem}{Theorem}[section]
\newtheorem{definition}[theorem]{Definition}
\newtheorem{lemma}[theorem]{Lemma}

\newtheorem{example}[theorem]{Example}

\usepackage[numbers,sort&compress]{natbib}
	\makeatletter %
	\renewcommand*{\NAT@spacechar}{~}
	\makeatother

	\ifx\usespringer\undefined
		
	\else
		
	\fi

\newcommand{\intWort}[1]{\emph{\textbf{#1}}}
\newcommand{\tuple}[1]{(#1)} %
\newcommand{\abs}[1]{\ensuremath{\left|#1\right|}} %
\newcommand{\bsq}[1]{\lq{#1}\rq} %
\newcommand{\menge}[1]{\ensuremath{\left\{#1\right\}}} %
\newcommand{\upgauss}[1]{\left\lceil#1\right\rceil} %
\newcommand{\jbracket}[1]{\ensuremath{\langle{#1}\rangle}} %

\newcommand{\UnaryOperatorS}[2][]{%
	\ifx&#1&%
	\ensuremath{\mathop{}\mathopen{}#2\mathopen{}}%
	\else%
	\ensuremath{\mathop{}\mathopen{}#2\mathopen{}\left(#1\right)}%
	\fi%
}
\newcommand{\UnaryOperator}[2][]{%
	\ifx&#1&%
	\ensuremath{\mathop{}\mathopen{}#2\mathopen{}}%
	\else%
	\ensuremath{\mathop{}\mathopen{}#2\mathopen{}(#1)}%
	\fi%
}
\newcommand{\Oh}[1]{\UnaryOperator[#1]{\mathcal{O}}}
\newcommand{\Om}[1]{\UnaryOperator[#1]{\Omega}}
\newcommand{\Ot}[1]{\UnaryOperator[#1]{\Theta}}

\usepackage[capitalise]{cleveref}
	\Crefname{algocf}{Algorithm}{Algorithms}
	\crefname{algocf}{Algo.}{Algos.}

\usepackage{xcolor}
	\definecolor{solarizedBase03}{HTML}{002B36}
	\definecolor{solarizedBase02}{HTML}{073642}
	\definecolor{solarizedBase01}{HTML}{586e75}
	\definecolor{solarizedBase00}{HTML}{657b83}
	\definecolor{solarizedBase0}{HTML}{839496}
	\definecolor{solarizedBase1}{HTML}{93a1a1}
	\definecolor{solarizedBase2}{HTML}{EEE8D5}
	\definecolor{solarizedBase3}{HTML}{FDF6E3}
	\definecolor{solarizedYellow}{HTML}{B58900}
	\definecolor{solarizedOrange}{HTML}{CB4B16}
	\definecolor{solarizedRed}{HTML}{DC322F}
	\definecolor{solarizedMagenta}{HTML}{D33682}
	\definecolor{solarizedViolet}{HTML}{6C71C4}
	\definecolor{solarizedBlue}{HTML}{268BD2}
	\definecolor{solarizedCyan}{HTML}{2AA198}
	\definecolor{solarizedGreen}{HTML}{859900}

\usepackage{forest}
\usepackage{bold-extra}

\usepackage{float}
\newfloat{algorithmflt}{t}{lop}
\makeatletter
\newcommand{\removelatexerror}{\let\@latex@error\@gobble}
\makeatother

\usepackage[group-separator={,},binary-units=true,exponent-product = \cdot, output-product = \cdot]{siunitx} %
\DeclareSIUnit{\NoUnit}{\relax}

\newcommand*{\UnaryBracketOperator}[2][]{%
    \ifx&#1&%
    \ensuremath{\mathop{}\mathopen{}#2\mathopen{}}%
    \else%
    \ensuremath{\mathop{}\mathopen{}#2\mathopen{}\left[#1\right]}%
    \fi%
}
\newcommand*{\twodots}{\mathbin{{.}\,{.}}}
\def\itemlabelstyle{} %

\usepackage[linesnumbered,ruled,vlined]{algorithm2e}
\SetArgSty{} %
\DontPrintSemicolon{}

\SetKwComment{Comment}{$\triangleright$\ }{}

\SetKwSty{MyKwFont}
\SetNlSty{bfseries}{\color{gray}}{}

\SetCommentSty{MyCommFont}

\DeclareMathOperator{\sort}{sort}
\DeclareMathOperator{\scan}{scan}

\usepackage{tikz,pgfplots}

\definecolor{plotcolor0}{RGB}{255,0,0}
\definecolor{plotcolor1}{RGB}{0,0,255}
\definecolor{plotcolor2}{RGB}{64,170,0}
\definecolor{plotcolor3}{RGB}{170,0,255}
\definecolor{plotcolor4}{RGB}{0,170,255}
\definecolor{plotcolor5}{RGB}{164,164,0}
\definecolor{plotcolor6}{RGB}{128,128,128}
\definecolor{plotcolor7}{RGB}{0,0,0}
\definecolor{plotcolor8}{RGB}{0,110,165}
\definecolor{plotcolor9}{RGB}{128,64,0}

\pgfplotscreateplotcyclelist{mycolor}{%
  plotcolor0, every mark/.append style={solid,scale=0.8}, mark=x \\%
  plotcolor1, every mark/.append style={solid,scale=0.6}, mark=o \\%
  plotcolor2, every mark/.append style={solid,scale=0.8}, mark=diamond \\%
  plotcolor3, every mark/.append style={solid,scale=0.8}, mark=triangle \\%
  plotcolor4, every mark/.append style={solid,scale=0.8,rotate=180}, mark=triangle \\%
  plotcolor5, every mark/.append style={solid,scale=0.65}, mark=square \\%
  plotcolor6, every mark/.append style={solid,scale=0.8}, mark=star \\%
  plotcolor7, every mark/.append style={solid,scale=0.8}, mark=+ \\%
  plotcolor8, every mark/.append style={solid,scale=0.8}, mark=pentagon \\%
  plotcolor9, every mark/.append style={solid,scale=0.9}, mark=Mercedes star flipped \\%
  plotcolor0, every mark/.append style={solid,scale=0.8}, mark=x, dashed \\%
  plotcolor1, every mark/.append style={solid,scale=0.6}, mark=o, dashed \\%
  plotcolor2, every mark/.append style={solid,scale=0.8}, mark=diamond, dashed \\%
  plotcolor3, every mark/.append style={solid,scale=0.8}, mark=triangle, dashed \\%
  plotcolor4, every mark/.append style={solid,scale=0.8,rotate=180}, mark=triangle, dashed \\%
  plotcolor5, every mark/.append style={solid,scale=0.65}, mark=square, dashed \\%
  plotcolor6, every mark/.append style={solid,scale=0.8}, mark=star, dashed \\%
  plotcolor7, every mark/.append style={solid,scale=0.8}, mark=+, dashed \\%
  plotcolor8, every mark/.append style={solid,scale=0.8}, mark=pentagon, dashed \\%
  plotcolor9, every mark/.append style={solid,scale=0.9}, mark=Mercedes star flipped, dashed \\%
  plotcolor0, every mark/.append style={solid,scale=0.8}, mark=x, densely dotted \\%
  plotcolor1, every mark/.append style={solid,scale=0.6}, mark=o, densely dotted \\%
  plotcolor2, every mark/.append style={solid,scale=0.8}, mark=diamond, densely dotted \\%
  plotcolor3, every mark/.append style={solid,scale=0.8}, mark=triangle, densely dotted \\%
  plotcolor4, every mark/.append style={solid,scale=0.8,rotate=180}, mark=triangle, densely dotted \\%
  plotcolor5, every mark/.append style={solid,scale=0.65}, mark=square, densely dotted \\%
  plotcolor6, every mark/.append style={solid,scale=0.8}, mark=star, densely dotted \\%
  plotcolor7, every mark/.append style={solid,scale=0.8}, mark=+, densely dotted \\%
  plotcolor8, every mark/.append style={solid,scale=0.8}, mark=pentagon, densely dotted \\%
  plotcolor9, every mark/.append style={solid,scale=0.9}, mark=Mercedes star flipped, densely dotted \\%
}

\pgfplotsset{
  compat=newest,
  major grid style={thin,dotted,color=black!50},
  minor grid style={thin,dotted,color=black!50},
  grid,
    cycle list name={mycolor},
  every axis/.append style={
    line width=0.5pt,
    tick style={
      line cap=round,
      thin,
      major tick length=4pt,
      minor tick length=2pt,
    },
    label style={
        font=\scriptsize
    }
  },
  title style={yshift=-2pt,font=\scriptsize},
  enlarge x limits=0.04,
  every tick label/.append style={font=\tiny},
  xlabel near ticks,
  ylabel near ticks,
  legend cell align=left,
  legend columns=3,
  legend pos=north east,
  legend style={
    font=\scriptsize,
    /tikz/every even column/.append style={column sep=3mm,black},
    /tikz/every odd column/.append style={black},
  },
  plcpplots/.style={
    width=0.33\textwidth,height=0.25\textwidth,
  },
  plcpplots4/.style={
    width=0.25\textwidth,height=0.25\textwidth,
  },
}

\makeatletter
\newcommand{\CustomLabel}[2]{%
   \protected@write \@auxout {}{\string \newlabel {#1}{{#2}{\thepage}{#2}{#1}{}} }%
   \hypertarget{#1}{#2}%
}
\makeatother

\def\vecFac{\ensuremath{\mathsf{factors}}}
\def\pqSplit{\ensuremath{\mathsf{PQSplit}}}
\def\vecChld{\ensuremath{\mathsf{requests}}}
\def\vecNChld{\ensuremath{\mathsf{nextRequests}}}
\def\vecPQ{\ensuremath{\mathsf{PQ}}}
\def\vecRes{\ensuremath{\mathsf{result}}}

\title{Bidirectional Text Compression in External Memory}
\author{Patrick Dinklage, Jonas Ellert, Johannes Fischer, \\Dominik K\"oppl, Manuel Penschuck}
\date{}

\begin{document}
\maketitle
\begin{abstract}
  Bidirectional compression algorithms work by substituting repeated substrings by references that, unlike in the famous LZ77-scheme, can point to either direction.\
  We present such an algorithm that is particularly suited for an external memory implementation.
  We evaluate it experimentally on large data sets of size up to 128 GiB (using only 16 GiB of RAM) and show that it is significantly faster than all known LZ77 compressors, while producing a roughly similar number of factors.
  We also introduce an external memory decompressor for texts compressed with any uni- or bidirectional compression scheme.
\end{abstract}

\section{Introduction}
Text compression is a fundamental task when storing massive data sets.
Most practical text compressors such as gzip, bzip2, 7zip, etc., scan a text file with a sliding window, replacing repetitive occurrences within this window.
Although this approach is memory and time efficient~\cite{bell86lz77,storer82lzss}, two occurrences of the same substring are neglected if their distance is longer than the sliding window.
More advanced solutions \cite[to mention only a few examples]{goto13lz,goto14lz,fischer18lz,kaerkkaeinen13lightweight} drop the idea of a sliding window, thereby finding also repetitions that are far apart in the text.
These so-called \iLZSS{}-algorithms have a better compression ratio in practice~\cite[Sect.~6]{ferragina13bit}.
In recent years, these algorithms have also been transformed to the \emph{external memory} (EM) model~\cite{karkkainen14parsing,kempa17lzendcompressed,belazzougui16decoding}.

In this article, we present a modification of \iLZSS{}, called \intWort{plcpcomp}, which is
based on the bidirectional compression scheme \iLcpcomp{} of~\citet{dinklage17tudocomp}, but is better suited for an efficient external memory implementation due to its memory access patterns.
We can compute this scheme by scanning the text and two auxiliary arrays stored in EM (one of them being the \emph{permuted longest common prefix array}, hence the acronym plcp)\@.
We underline the performance of our algorithm with evaluations showing that it is faster than any known \iLZSS{} compressor for massive non-highly repetitive data sets.
We also present the first external decompressor for files that are compressed with a bidirectional scheme.

\subsection{Related Work}
Our work is the first to join the fields of bidirectional and external memory compression. 

\subsubsection{Bidirectional Schemes}
First considerations started with~\citet{storer82lzss} who also coined this notation.
\citet{dissGallant} proved that finding the \emph{optimal} bidirectional parsing, i.e., a bidirectional parsing with the lowest number of factors, is NP-complete.
\citet{dinklage17tudocomp} were the first to present a greedy algorithm for producing a bidirectional parsing called \iLcpcomp{}, which performs well in practice, but comes with no theoretical performance guarantees on its size.
\citet{mauer17lfs} combined the techniques for \iLcpcomp{}~\cite{dinklage17tudocomp} and the longest-first grammar compression~\cite{nakamura09lfs} in a compression algorithm running in \Oh{n^2} time, 
which was subsequently improved to \Oh{n \lg n} time by~\citet{nishi18lzlfs}.
Recently,~\citet{gagie18approxlz} showed an upper bound of $z = \Oh{b \lg(n/b)}$ and a lower bound of $z = \Om{b \lg n}$ for some specific strings, where $b$ and $z$ denote the minimal number of factors in an optimal \emph{bidirectional} parsing and in an optimal \emph{unidirectional} parsing, respectively.
This implies that bidirectional parsing can be exponentially better than unidirectional parsing.
They also proposed a bidirectional parsing based on the Burrows-Wheeler transform (BWT).
\citet{kempa18stringattractors} introduced so-called \emph{string attractors},
showed that a bidirectional scheme is a string attractor and that every string attractor can be represented with a bidirectional scheme.
Last but not least, the bidirectional scheme of~\citet{nishimoto19lzrr} guarantees to produce at most as many factors as \iLZSS{},
but has the disadvantage of a super-quadratic running time.

\subsubsection{EM Compression Algorithms}
Yanovsky~\cite{yanovsky11recoil} presented a compressor called \emph{ReCoil} that is specialized on large DNA datasets.
Ferragina~et~al.~\cite{ferragina12bwt} gave a construction algorithm of the Burrows-Wheeler transform in EM. %
For \iLZSS{} compression,~\citet{karkkainen14parsing} devised two algorithms called \iLZscan{} and \iEMLPF{}\@.
The former performs well on highly-repetitive data, but gets outperformed easily by \iEMLPF{} on other kinds of datasets.
The \iLZSS{} compressed files can be decompressed with an algorithm due to~\citet{belazzougui16decoding}, which also works in general for all
files that have been compressed by a unidirectional scheme.
Finally,~\citet{kempa17lzendcompressed} presented an EM algorithm for computing the \iLZend{} scheme~\cite{kreft10lzend}, a variant of \iLZSS{}.

\subsection{Preliminaries}
\subparagraph*{Model of computation}\label{par:model_of_compute}
We use
the commonly accepted EM model by Aggarwal and Vitter~\cite{aggarwal88iomodel}.
It features two memory types, namely fast internal memory~(IM) which may hold up to $M$ data words, and slow EM of unbounded size.
The measure of the performance of an algorithm is the number of input and output operations (I/Os) required, where each I/O transfers a block of $B$ consecutive words between memory levels.
Reading or writing $n$ contiguous words from or to disk requires $\scan(n) = \Theta(n/B)$~I/Os.
Sorting $n$ contiguous words requires $\sort(n)=\Theta((n/B) \cdot \log_{M/B}(n/B))$~I/Os.
For realistic values of $n$, $B$, and $M$, we stipulate that $\scan(n) < \sort(n) \ll n$.

\subparagraph*{Text}
Let $\Sigma$ denote an integer alphabet of size $\sigma = \abs{\Sigma} = n^{\Oh{1}}$ for a natural number~$n$.
The alphabet~$\Sigma$ induces the \intWort{lexicographic order}~{$\prec$} on the set of strings~$\Sigma^*$.
Let $\abs{T}$ denote the length of a string~$T \in \Sigma^*$.
We write $T[j]$ for the $j$-th character of $T$, where $1 \le j \le n$.
Given $T \in \Sigma^*$ consists of the concatenation $T = UVW$ for $U,V,W \in \Sigma^*$,
we call $U$, $V$, and $W$ a \intWort{prefix}, a \intWort{substring}, and a \intWort{suffix} of $T$, respectively.
Given that the substring~$V$ starts at the $i$-th and ends at the $j$-th position of $T$, we also
write $V = T[i\twodots j]$ and $W = T[j+1\twodots]$.
In the following, we take an element $T \in \Sigma^*$ with $\abs{T} = n$, and call it \intWort{text}.
We stipulate that $T$ ends with a sentinel $T[n] = \texttt{\$} \not\in \Sigma$ that is lexicographically smaller than every character of~$\Sigma$.

\subparagraph*{Text Data Structures}
Let $\SA$ denote the \intWort{suffix array}~\cite{manber93sa} of $T$.
The entry~$\SA[i]$ is the starting position of the $i$-th lexicographically smallest suffix such that $T[\SA[i] \twodots] \prec T[\SA[i+1] \twodots]$ for all integers~$i$ with $1 \le i \le n-1$.
Let $\ISA$ of $T$ be the inverse of $\SA$, i.e., $\ISA[\SA[i]] = i$ for every $i$ with $1 \le i \le n$.
The \intWort{Burrows-Wheeler transform (BWT)}~\cite{burrows94bwt} of~$T$ is the string~\BWT{} with $\BWT[i] = T[n]$ if $\SA[i] = 1$ and $\BWT[i] = T[\SA[i]-1]$ otherwise, for every~$i$ with $1 \le i \le n$.
The \intWort{LCP array} is an array with the property that $\LCP[i]$ is the length of the longest common prefix (LCP) of $T[\SA[i] \twodots]$ and $T[\SA[i-1] \twodots]$ for $i=2,\ldots,n$.
For convenience, we stipulate that $\LCP[1] := 0$.
The array~$\Phi$ is defined as $\Phi[i] := \SA[\ISA[i]-1]$, and $\Phi[i] := n$ in case that $\ISA[i] = 1$.
The \intWort{PLCP array} $\PLCP$ stores the entries of $\LCP$ in text order, i.e., $\PLCP[\SA[i]] = \LCP[i]$.
\cref{figPrelimtextDS} illustrates the introduced data structures.

  \begin{figure}%
    \centering{%
  \setlength{\tabcolsep}{0.38em}%
        \begin{tabular}{l*{23}{c}}
                    \toprule
$i$    & 1  & 2  & 3  & 4  & 5  & 6  & 7  & 8  & 9  & 10 & 11 & 12 & 13 & 14 & 15 & 16 & 17 & 18 & 19 & 20 & 21 & 22 \\
\midrule
$T$     & a  & b  & a  & b  & b  & a  & b  & a & b  & a  & b  & b  & a  & b  & b  & a  & a  & b  & a  & b  & a  & \$ \\
\SA{}   & 22 & 21 & 16 & 19 & 17 & 6  & 1  & 8 & 13 & 3  & 10 & 20 & 15 & 18 & 5  & 7  & 12 & 2  & 9  & 14 & 4  & 11 \\
\ISA{}  & 7  & 18 & 10 & 21 & 15 & 6  & 16 & 8 & 19 & 11 & 22 & 17 & 9  & 20 & 13 & 3  & 5  & 14 & 4  & 12 & 2  & 1 \\
$\Phi$    & 6  & 12 & 13 & 14 & 18 & 17 & 5  & 1 & 2  & 3  & 4  & 7  & 8  & 9  & 20 & 21 & 19 & 15 & 16 & 10 & 22 & 11\\
\LCP{}  & 0  & 0  & 1  & 1  & 3  & 5  & 4  & 7 & 2  & 4  & 5  & 0  & 2  & 2  & 4  & 5  & 3  & 5  & 6  & 1  & 3  & 4 \\
\PLCP{} & 4  & 5  & 4  & 3  & 4  & 5  & 5  & 7 & 6  & 5  & 4  & 3  & 2  & 1  & 2  & 1  & 3  & 2  & 1  & 0  & 0  & 0 \\
\bottomrule
        \end{tabular}
    }%
    \caption{Suffix array, its inverse, $\Phi$, LCP array, and PLCP array of our running example string~$T$.} %
    \label{figPrelimtextDS}
\end{figure}

\subparagraph*{Idea for Using PLCP for Compression} Given a suffix~$T[i\twodots]$ starting at text position~$i$, $\PLCP[i]$ is the length of the longest common prefix of
this suffix and the suffix~$T[\Phi[i]\twodots]$, which is its lexicographical predecessor among all suffixes of $T$.
The longest common prefix of these two suffixes~$T[i\twodots]$ and~$T[\Phi[i]\twodots]$ is $T[i \twodots i+\PLCP[i]-1]$.
The longest string among all these longest common prefixes (for each $i$ with $1 \le i \le n$) is one of the longest re-occurring substrings in the text.
Finding this longest re-occurring substring with \PLCP{} and $\Phi$ is the core idea of our compression algorithm.
This algorithm produces a bidirectional scheme, which is defined as follows.

\section{Compression Scheme}
\label{sec:compression-scheme}

A \intWort{bidirectional scheme}~\cite{storer82lzss} is defined by a factorization $F_1 \cdots F_b = T$ of a text~$T$.
A factor~$F_x$ is either a \intWort{referencing} factor or a \intWort{literal} factor.
A referencing factor~$F_x$ is associated with a pair~$(\src, \ell)$ such that
$F_x$ and $T[\src\twodots\src+\ell-1]$ are two different but possibly overlapping occurrences of the substring~$F_x$ in~$T$.
The pair~$(\src,\ell)$ and the text position~$\src$ are called \intWort{reference} and \intWort{referred position}, respectively.
A factorization is \intWort{cycle-free}, i.e., references are not allowed to have cyclic dependencies.
A factorization is called \intWort{$\threshold$-restricted} for an integer $\threshold \ge 2$
if each referencing factor $F_x$ is at least $\threshold$ characters long (i.e., $\ell \ge \threshold$).

A \intWort{unidirectional scheme} is a special case of a bidirectional scheme,
with the restriction that the referred position of a referencing factor~$F_x$ must be smaller than the starting position of~$F_x$.
The most prominent example of a unidirectional scheme is the \iLZSS{} factorization, whose factorization is usually designed to be 2-restricted.

\subsection{Coding}

A bidirectional scheme codes the factors by substituting referencing factors with their associated references while keeping literal factors as strings.
By doing so, the coding is a list whose $x$-th element is either a string (corresponding to a literal factor) or a reference representing the $x$-th factor ($1 \le x \le b$), which is referencing.

\begin{figure}
	\centering
	\includegraphics[width=0.5\linewidth]{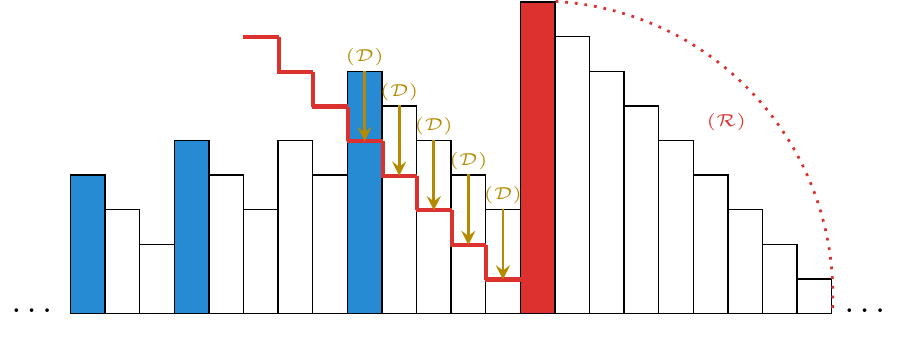}
	\caption{Visualization of Rules~\ref{itPLCPdecrease} and~\ref{itPLCPreplace} being applied. 
		Bars represent $\PLCP$ values. }
	\label{figRuleVis}
\end{figure}

\begin{figure}
	\setlength{\tabcolsep}{0.38em}%
	\newcommand{\XR}[1]{{\color{solarizedRed}\underline{#1}}}%
	\newcommand{\X}[1]{{\color{solarizedGreen}\noindent\fbox{#1}}}%
	\newcommand{\XD}[1]{{\color{solarizedBlue}\underline{#1}}}%
	\centering{\scalebox{0.95}{%
			\begin{tabular}{l*{23}{c}}
				\toprule
				$i$     & 1      & 2     & 3      & 4      & 5      & 6      & 7      & 8     & 9      & 10     & 11     & 12     & 13     & 14     & 15    & 16     & 17 & 18 & 19 & 20 & 21 & 22 \\
				$T$    & a      & b     & a      & b      & b      & a      & b      & a     & b      & a      & b      & b      & a      & b      & b     & a      & a  & b  & a  & b  & a  & \$ \\
				\midrule
				$\PLCP{}$ & 4      & 5     & 4      & 3      & 4      & 5      & 5      & 7     & 6      & 5      & 4      & 3      & 2      & 1      & 2     & 1      & 3  & 2  & 1  & 0  & 0  & 0 \\
				$\PLCP{}^1$ & 4      & 5     & 4      & 3      & \XD{3} & \XD{2} & \XD{1} & \X{0} & \XR{0} & \XR{0} & \XR{0} & \XR{0} & \XR{0} & \XR{0} & 2     & 1      & 3  & 2  & 1  & 0  & 0  & 0 \\
				$\PLCP{}^2$ & \XD{1} & \X{0} & \XR{0} & \XR{0} & \XR{0} & \XR{0} & 1      & 0     & 0      & 0      & 0      & 0      & 0      & 0      & 2     & 1      & 3  & 2  & 1  & 0  & 0  & 0 \\
				$\PLCP{}^4$ & 1      & 0     & 0      & 0      & 0      & 0      & 1      & 0     & 0      & 0      & 0      & 0      & 0      & 0      & 2     & 1      & \X{0}  & \XR{0}  & \XR{0}  & 0  & 0  & 0\\
				$\PLCP{}^3$ & 1      & 0     & 0      & 0      & 0      & 0      & 1      & 0     & 0      & 0      & 0      & 0      & 0      & 0      & \X{0} & \XR{0} & 0 & 0  & 0  & 0  & 0  & 0 \\
				\bottomrule
			\end{tabular}
	}}
	\caption{%
		Step-by-step computation of the instructions in \cref{sec:compression-scheme} computing the \protect\iPlcpcomp{} compression scheme 
		on $T = \exampleString$.
		We overwrite values of \PLCP{} according to Rules~\ref{itPLCPdecrease} and~\ref{itPLCPreplace}.
		Each row $\PLCP^{i}$ shows \PLCP{} after creating the $i$-th referencing factor starting at a position whose \PLCP{} entry is surrounded by a box.
		Changed entries according to Rules~\ref{itPLCPdecrease} and~\ref{itPLCPreplace} are underlined.
	}
	\label{figPLCPSchemeExample}
\end{figure}

The \iPlcpcomp{} scheme and its predecessor, the \iLcpcomp{} scheme~\cite{dinklage17tudocomp}, are bidirectional schemes.
Both schemes are greedy, as they create a referencing factor equal to the longest re-occurring substring of the text that is not yet part of a factor.
They differ in the selection of such a substring in case that there are multiple candidates with the same length.
The \iPlcpcomp{} scheme can be computed with a rewritable \PLCP{} array and the following instructions:
\begin{enumerate} 
    \item Compute the set of \intWort{candidate positions} $\Cand := \{i \mid \PLCP[i] \ge \PLCP[j]$ for all text positions~$j\}$.
    \item Let $\dst$ be the leftmost position of all candidate positions~$\Cand$.
    Terminate if $\PLCP[\dst] < \threshold$.
    \item Create a referencing factor by replacing $T[\dst\twodots\dst+\PLCP[\dst]-1]$ with the reference $(\Phi[\dst],$ $\PLCP[\dst])$
    \item Apply the following rules to ensure that we do not create overlapping factors (cf.\ \cref{figRuleVis}):
        \begin{itemize} 
            \item[\CustomLabel{itPLCPdecrease}{(\ensuremath{\mathcal{D}})}]
              Decrease $\PLCP[j] \gets \min\tuple{\PLCP[j], \dst-j}$ for every $j \in [\dst-\PLCP[\dst], \dst)$.
            \item[\CustomLabel{itPLCPreplace}{(\ensuremath{\mathcal{R}})}]
              Remove the factored positions by setting $\PLCP[\dst+k] \gets 0$ for every $k \in [0,$ $\PLCP[\dst])$.
        \end{itemize}
    \item Recurse with the modified $\PLCP$.
\end{enumerate}

\noindent An application of the above instructions on our running example is given in \cref{figPLCPSchemeExample}.
The coding is visualized in \cref{figPLCPfactorization}.
There and in the following figures, we fix $\threshold := 2$.

\begin{figure}
    \centering{%
        \begin{adjustbox}{valign=c}%
        \includegraphics[width=0.55\textwidth]{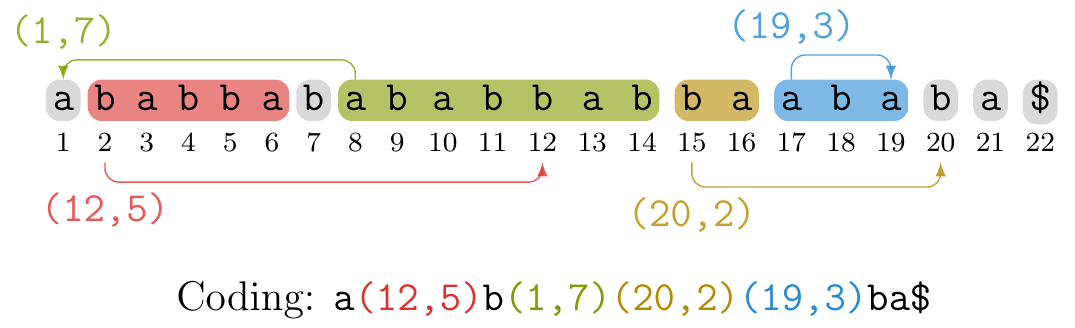}%
        \end{adjustbox}
        \hfill
        \begin{adjustbox}{valign=c}%
 \begin{tabular}{lllll}
\toprule
$\dst$ & $\src = \Phi[\dst]$ & \textrm{length} \\
\midrule
8   & 1                & 7 \\
2   & 12               & 5 \\
17  & 19               & 3 \\
15  & 20               & 2 \\
\bottomrule
 \end{tabular}
        \end{adjustbox}
    }%
    \caption{Coding of \protect\iPlcpcomp{} with $\threshold = 2$.
        The factorization described in \cref{figPLCPSchemeExample} computes four referencing factors, listed in the table on the right.
        These factors are coded by their references.
        The factorization with \PLCP{} in \cref{figPLCPSchemeExample} already determines the starting position and the lengths of all referencing factors
        (columns \bsq{$\dst$} and \bsq{\text{length}} in the table).
        The referred positions are obtained using $\Phi$ (column \bsq{$\src$} in the table).
        The figure on the left illustrates factors as boxed substrings and
        the references as arrows from the starting positions of referencing factors to their respective referred positions.
    }
    \label{figPLCPfactorization}
 \end{figure}

\subsection{Comparison to \texorpdfstring{\protect\iLcpcomp{}}{lcpcomp}}
The difference to \iLcpcomp{}~\cite{dinklage17tudocomp} is that we fix $\dst$ to be the \emph{leftmost} of all candidate positions in~$\Cand$.
\citet{dinklage17tudocomp} presented an algorithm computing the \iLcpcomp{} scheme in \Oh{n \lg n} time with a heap storing the candidate positions
ranked by their \PLCP{} values.
We can adapt this algorithm to compute the \iPlcpcomp{} scheme by altering the order of the heap to rank the candidate positions first by
their \PLCP{} values (maximal \PLCP{} values first) and second (in case of equal \PLCP{} values) by their values themselves (minimal text positions first).

Since \iLcpcomp{} is cycle-free~\cite[Lemma 4]{dinklage17tudocomp} regardless of the selection of $\dst \in \Cand$,
we conclude that \iPlcpcomp{} is also cycle-free, i.e., the substitution of substrings by references is reversible.

\section{Computing the Factorization without Random Access}\label{secPLCPcompAlgo}
In this section, we present an algorithm for computing the \iPlcpcomp{} scheme, 
which linearly scans \PLCP{} without changing its contents.
Instead of maintaining a heap storing all text positions ranked by their \PLCP{} values,
we compute the factorization by scanning the text sequentially from left to right.
Although the algorithm will produce the \iPlcpcomp{} factorization, it does not compute it in the order explained previously (starting with the longest factor).
Instead, it first determines a subset of those substrings that define a referencing factor according to the \iPlcpcomp{} scheme.
The starting positions of these substrings have a \PLCP{} value that is relatively large compared to their
neighboring positions. We call those starting positions \emph{peaks}.

Formally, we call a text position $\dst$ a \intWort{peak} if $\PLCP[\dst] \ge \threshold$ and one of the following conditions holds:
\begin{enumerate}
        \item  $\dst = 1$,
        \item  $\PLCP[\dst-1] < \PLCP[\dst]$,\footnote{A subset of the so-called \emph{irreducible} \PLCP{} entries~\cite[Lemma 4]{karkkainen09plcp} have this property.} or
    \item  there is a referencing factor ending at $\dst-1$.
\end{enumerate}

    A peak~$\dst$ is called \intWort{interesting} if there is no text position~$j$ with $\dst \in (j, j+\PLCP[j])$ 
    and $\PLCP[j] \ge \PLCP[\dst]$.
    An interesting peak~$\dst$ is called \intWort{maximal} if there is no interesting peak~$j$ with $j \in (\dst, \dst+\PLCP[\dst])$.

Given an interesting peak~$\dst$, there is no text position~$j$ with $\PLCP[j] \ge \PLCP[\dst]$ that becomes the starting position of a referencing factor containing $T[\dst]$ (such that $\PLCP[\dst]$ cannot be removed according to Rule~\ref{itPLCPreplace}).
Given a maximal peak~$\dst$, there is additionally no text position~$j$ with $\PLCP[j] > \PLCP[\dst]$ for which we apply Rule~\ref{itPLCPdecrease} on $\PLCP[\dst]$ after factorizing $T[j\twodots{}j+\PLCP[j]-1]$.
Informally, we can determine whether a peak is interesting by looking at the \PLCP{} values \emph{before} this peak,
whereas we need to also look \emph{ahead} for determining whether a peak is maximal.
Given that there is at least one \PLCP{} entry with a value of at least $\threshold$, we can find a maximal peak, since
the leftmost position $\min \menge{i \in [1\twodots{}n] \mid \PLCP[i] \ge \PLCP[j] \text{~for all~} j \text{~with~} 1 \le j \le n}$ 
among all positions with the highest \PLCP{} value is a maximal peak.
The following lemma states that we can always factorize the leftmost maximal peak, regardless of whether the text has even higher peaks.

\begin{lemma}\label{lemGoodPeak}
    If the text position~$\dst$ is a maximal peak, then
        $T[\dst \twodots \dst+\PLCP[\dst]-1]$ is a referencing factor.
\end{lemma}
\begin{proof}
When applying Rules~\ref{itPLCPreplace} and~\ref{itPLCPdecrease},
    we do not change the value of $\PLCP[\dst]$, since $\dst$ is a maximal peak.
    Therefore, we will eventually create a referencing factor starting with $\dst$.
\end{proof}

\noindent Our preliminary algorithm consists of the following steps:
\begin{enumerate} 
    \item Scan \PLCP{} for the leftmost maximal peak~$\dst$. \label{itAlgoSearchPeak}
    \item Terminate if no such peak exists.
    \item Create the referencing factor $T[\dst \twodots \dst+\PLCP[\dst]-1]$.
    \item Apply Rules~\ref{itPLCPreplace} and~\ref{itPLCPdecrease}. \label{itAlgoApplyRules}
    \item Interpret $T[1 \twodots \dst-1]$ and $T[\dst+\PLCP[\dst] \twodots n]$ as two independent strings and recurse on each of them individually.
\end{enumerate}

\begin{algorithmflt}
	\begin{algorithm}[H]
		$L \gets \emptyset{}$ \tcp*[r]{Step 1\texttt a}
		\For(\tcp*[f]{Step 1\texttt b}){$\dst = 1$ to $n$}{%
			\If(\tcp*[f]{Step 2\phantom{\texttt a}}){$\dst$ is a maximal peak}{%
				create a referencing factor replacing $T[\dst \twodots \dst+\PLCP[\dst]-1]$ \tcp*[r]{Step 3\phantom{\texttt a}}
				apply Rule~\ref{itPLCPdecrease} to the peaks in $L$ \;
				\While{$L$ contains maximal peaks}{\label{lineAlgStartProcessL}
					$j \gets $ rightmost maximal peak in $L$ \;
					create referencing factor replacing $T[j \twodots j+\PLCP[j]-1]$ \;
					apply Rules~\ref{itPLCPdecrease} and~\ref{itPLCPreplace} to the peaks in $L$ \;
					remove those elements of $L$ that are no longer interesting peaks \;
				} \label{lineAlgEndProcessL}
				$\dst \gets \dst+\PLCP[\dst]$ \;
			}
			\If{$\dst$ is an interesting peak}{%
				$L \gets L \cup \menge{\dst}$ \;
			}
		}
		\caption{Computation of \protect\iPlcpcomp{} factors.}
		\label{algoRoutineAPrime}
	\end{algorithm}
\end{algorithmflt}

This algorithm produces the \iPlcpcomp{} scheme, because
\begin{itemize}
    \item $T[\dst \twodots \dst+\PLCP[\dst]-1]$ is a referencing factor for each selected leftmost maximal peak~$\dst$ according to \cref{lemGoodPeak}, and
    \item the part $T[1 \twodots \dst-1]$ can be factorized independently from how $T[\dst+\PLCP[\dst] \twodots]$ is factorized, and vice versa.
       That is because, having already $T[\dst\twodots\dst+\PLCP[\dst]-1]$ factorized, we can no longer create a factor that covers a text position in the range~$[\dst\twodots\dst+\PLCP[\dst]-1]$.
\end{itemize}
Hence, we can factorize~$T[1 \twodots \dst-1]$ without considering the factorization of the rest of the text to produce the correct \iPlcpcomp{} scheme.
\Cref{figPLCPCompExample} illustrates the computation of the \iPlcpcomp{} factorization with this algorithm.

\begin{figure}
  \setlength{\tabcolsep}{0.38em}%
\newcommand{\XR}[1]{{\color{solarizedRed}\underline{#1}}}%
\newcommand{\X}[1]{{\color{solarizedGreen}\noindent\fbox{#1}}}%
\newcommand{\XD}[1]{{\color{solarizedBlue}\underline{#1}}}%
    \centering{\scalebox{0.98}{%
        \begin{tabular}{l*{23}{c}}
            \toprule
$i$     & 1      & 2     & 3      & 4      & 5      & 6      & 7      & 8     & 9      & 10     & 11     & 12     & 13     & 14     & 15    & 16     & 17 & 18 & 19 & 20 & 21 & 22 \\
 $T$    & a      & b     & a      & b      & b      & a      & b      & a     & b      & a      & b      & b      & a      & b      & b     & a      & a  & b  & a  & b  & a  & \$ \\
\midrule
$\PLCP{}$ & 4      & 5     & 4      & 3      & 4      & 5      & 5      & 7     & 6      & 5      & 4      & 3      & 2      & 1      & 2     & 1      & 3  & 2  & 1  & 0  & 0  & 0 \\
$\PLCP{}^1$ & \XD{1}      & \X{0}     & \XR{0}      & \XR{0}      & \XR{0}      & \XR{0}      & 5      & 7     & 6      & 5      & 4      & 3      & 2      & 1      & 2     & 1      & 3  & 2  & 1  & 0  & 0  & 0 \\
$\PLCP{}^1$ & {1}      & {0}     & {0}      & {0}      & {0}      & {0}      & \XD{1}      & \X{0}     & \XR{0}      & \XR{0}      & \XR{0}      & \XR{0}     & \XR{0}      & \XR{0}  & 2     & 1      & 3  & 2  & 1  & 0  & 0  & 0 \\
$\PLCP{}^2$ & {1}      & {0}     & {0}      & {0}      & {0}      & {0}      & {1}      & {0}     & {0}      & {0}      & {0}      & {0}     & {0}      & {0}  & \X{0}     & \XR{0}      & 3  & 2  & 1  & 0  & 0  & 0 \\
$\PLCP{}^3$ & {1}      & {0}     & {0}      & {0}      & {0}      & {0}      & {1}      & {0}     & {0}      & {0}      & {0}      & {0}     & {0}      & {0}  & {0}     & {0}      & \X{0}  & \XR{0}  & \XR{0}  & 0  & 0  & 0 \\
\bottomrule
        \end{tabular}
    }}
    \caption{%
      Step-by-step computation of our \protect\iPlcpcomp{} algorithm 
      on $T = \exampleString$.
        While the instructions of the scheme (cf.\ \cref{sec:compression-scheme}) always replace the factor starting at a position with the maximal \PLCP{} value (cf.\ \cref{figPLCPSchemeExample}),
        our algorithm described in \cref{secPLCPcompAlgo} creates a factor at the leftmost maximal peak.
        Our algorithm computes the same factorization as described in the \protect\iPlcpcomp{} scheme, but in different order.
    }
    \label{figPLCPCompExample}
\end{figure}

However, as the algorithm overwrites entries of \PLCP{}, it is not yet satisfying.
A rewritable \PLCP{} array would have to be kept in RAM, costing us $n \lg n$ bits of space if we require constant time read and write access.
Instead of keeping $\PLCP[1 \twodots \dst-1]$ in RAM, we now show that it suffices to manage only the \PLCP{} values of the \emph{interesting peaks}.
For that, we enhance the search of the leftmost maximal peak by replacing the first step of the algorithm by the following instructions:
\begin{enumerate} 
    \item[\itemlabelstyle \ref*{itAlgoSearchPeak}a.] Create an empty list of peaks~$L$.
    \item[\itemlabelstyle \ref*{itAlgoSearchPeak}b.] Scan $T$ from left to right until a maximal peak $\dst$ is found. While doing so, insert all visited interesting peaks into~$L$.
\end{enumerate}
Another alternation is that we apply Step~\ref{itAlgoApplyRules} only to the peaks stored in~$L$.
There, we scan $L$ from right to left while applying Rule~\ref{itPLCPdecrease} and removing all elements that are no longer interesting peaks.
The modified algorithm is sketched as pseudo code in \cref{algoRoutineAPrime}.

\begin{example}\label{exPlcpAlgo}
\Cref{figPLCPAlgoExample} illustrates \cref{algoRoutineAPrime} on the prefix $T[1\twodots 14] = \texttt{ababbabababbab}$ of our running example in three steps.
The peaks at positions~$1$~and~$2$ are interesting.
Since the peak at position~$2$ is the highest interesting peak, it is the maximal peak, which is detected after scanning~$\PLCP[1\twodots 6]$ (\cref{figPLCPAlgoExample1}).
In the second step (\cref{figPLCPAlgoExample2}), the referencing factor~$F_1$ is introduced, which starts at this maximal peak.
As a consequence, Rule~\ref{itPLCPdecrease} is applied to the only peak stored in~$L$, the one at position~1. However, because the \PLCP{} value~$1$ is below the threshold~$\threshold=2$, the peak at position~1 is removed from $L$.
Since $L$ is then empty, we proceed with the next scan for a maximal peak starting from position~7.
By definition, the peak at position~7 becomes interesting. The next maximal peak is detected at position~8 (\cref{figPLCPAlgoExample3}).
The factor~$F_2$ (\cref{figPLCPAlgoExample4}) is introduced, and Rule~\ref{itPLCPdecrease} is applied to the peak at position~7.
Its \PLCP{} value drops below our threshold and thus it is removed from $L$.
Finally, the prefix $T[1\twodots 14]$ has been processed.
\end{example}

In \cref{algoRoutineAPrime}, we omit all other peaks that are not stored in $L$ when applying Rules~\ref{itPLCPdecrease} and~\ref{itPLCPreplace}).
Thus, it suffices to maintain the \PLCP{} value of each peak in $L$ in an extra list instead of maintaining a complete rewritable \PLCP{} array.
In the following, we prove why this omission still produces the correct factorization (\cref{lemValidPlcpcomp}).
For that, we show that we can produce the \iPlcpcomp{} factors contained in $T[1 \twodots \dst+\PLCP[\dst]-1]$ only with the \PLCP{} values of the peaks stored in $L$ (first recursive call).
We start with the following property of $L$:

\begin{lemma}\label{lemLAscendingOrder}
	The positions stored in $L$ are in strictly ascending order with respect to their LCP values.
\end{lemma}
\begin{proof}
	Let $\dst$ be the leftmost maximal peak.
	Assume that there is an entry~$L[k] < \dst$ with $1 \le k \le \abs{L-1}$ and $\PLCP[L[k+1]] \le \PLCP[L[k]]$.
	Since $L[k]$ is an interesting peak,
	there is no text position~$j$ with $L[k] \in (j, j+ \PLCP[j])$ and $\PLCP[j] \ge \PLCP[L[k]]$.
	Since $L[k+1]$ is the succeeding interesting peak of $L[k]$ (with respect to text order),
	$L[k+1] < L[k] + \PLCP[L[k]]$ must hold.
	Otherwise, $L[k]$ would be a maximal peak, which contradicts the fact that $\dst$ is the leftmost maximal peak.
	However, the condition $\PLCP[L[k]] < L[k+1]$ must hold for $L[k+1]$ to be interesting.
\end{proof}
Next, we examine the result of creating the referencing factor $T[\dst \twodots \dst+\PLCP[\dst]-1]$ starting at the maximal peak~$\dst$.
After creating this factor, the \PLCP{} values of peaks near $\dst$ can be decreased.
However, this causes at most one new peak as can be seen by the following lemma:

\begin{lemma}\label{lemNoNewPeaks}
	Applying Rules~\ref{itPLCPdecrease} and~\ref{itPLCPreplace} after creating a referencing factor~$F_x$
	does not cause new peaks, with the only possible exception of the position succeeding the end of~$F_x$.
\end{lemma}
\begin{proof}
	Let $\dst$ be the starting position of the referencing factor~$F_x$ and
	let $j < \dst$ be a position that is not a peak at the time before the creation of $F_x$.
	Then $\PLCP[j-1] \ge \PLCP[j]$.
	After creating~$F_x$, it holds that
	\[
	\PLCP'[j-1] = \min\tuple{\PLCP[j-1], \dst - j} \ge \min\tuple{\PLCP[j], \dst - j - 1} = \PLCP'[j],
	\]
	where $\PLCP'$ is the modified \PLCP{} array after applying Rules~\ref{itPLCPdecrease} and~\ref{itPLCPreplace}.
	Hence, position~$j$ did not become a peak.
	If $j = \dst + \PLCP[\dst]$ is the position succeeding the end of~$F_x$, then $\PLCP[\dst+\PLCP[\dst]-1]=0$ according to Rule~\ref{itPLCPreplace}.
	Hence, $j$ becomes a peak if $\PLCP[j] \ge \threshold > 0$.
\end{proof}

\begin{figure}[t]
    \centering
    \begin{subfigure}[b]{0.45\textwidth}
        \includegraphics[width=\textwidth]{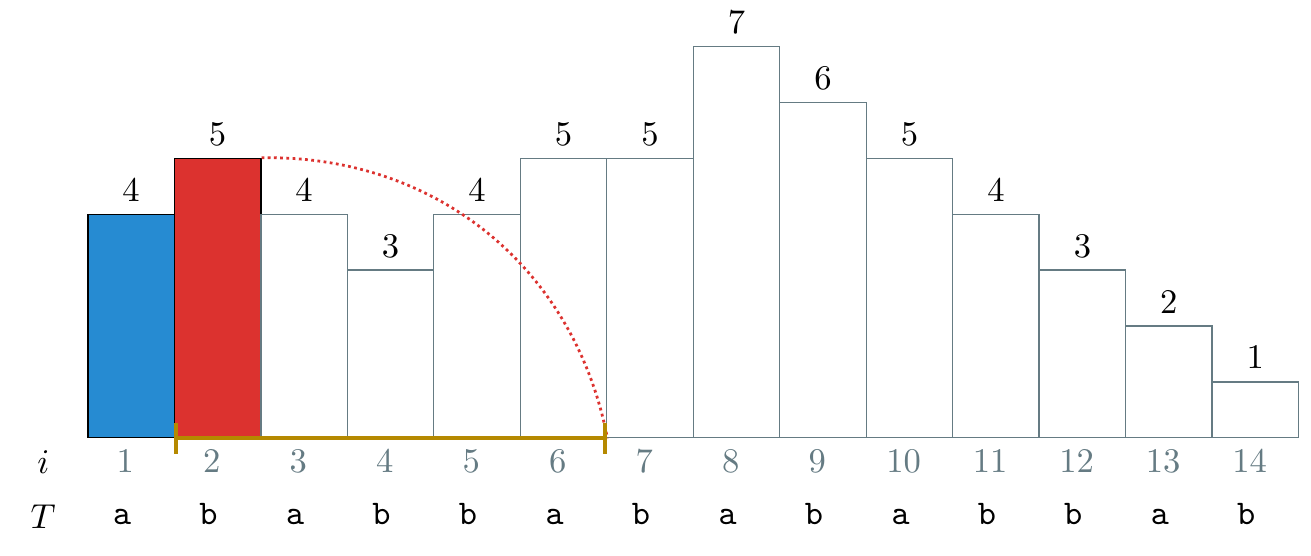}
        \caption{A maximal peak has been detected at $i=2$, an interesting peak is at $i=1$.}
        \label{figPLCPAlgoExample1}
    \end{subfigure}
    \hfill
    \begin{subfigure}[b]{0.45\textwidth}
        \includegraphics[width=\textwidth]{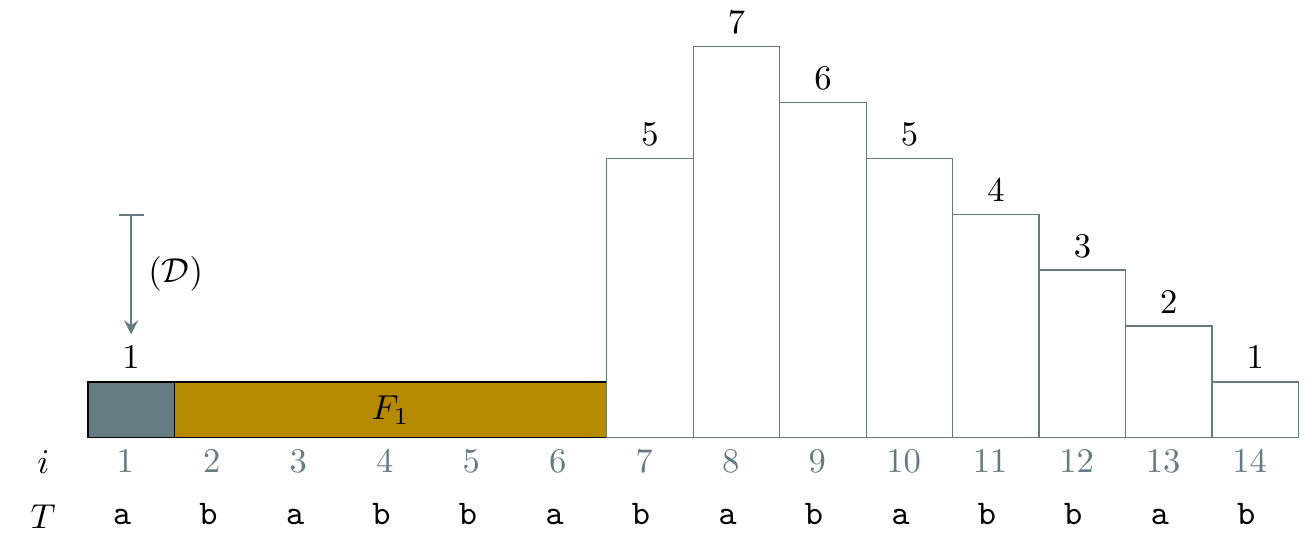}
        \caption{The referencing factor $F_1$ is introduced and Rule~\ref{itPLCPdecrease} is applied to the peak at $i=1$.}
        \label{figPLCPAlgoExample2}
    \end{subfigure}

    \vspace{2ex}
    \begin{subfigure}[b]{0.45\textwidth}
        \includegraphics[width=\textwidth]{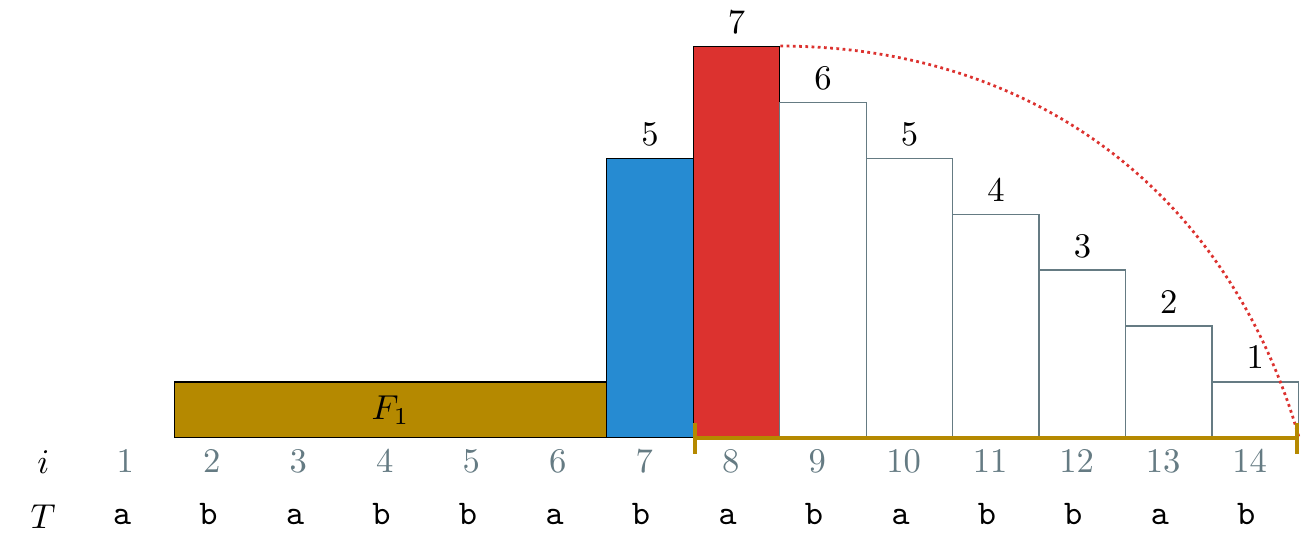}
        \caption{A maximal peak has been detected at $i=8$, an interesting peak is at $i=7$.}
        \label{figPLCPAlgoExample3}
    \end{subfigure}
    \hfill
    \begin{subfigure}[b]{0.45\textwidth}
        \includegraphics[width=\textwidth]{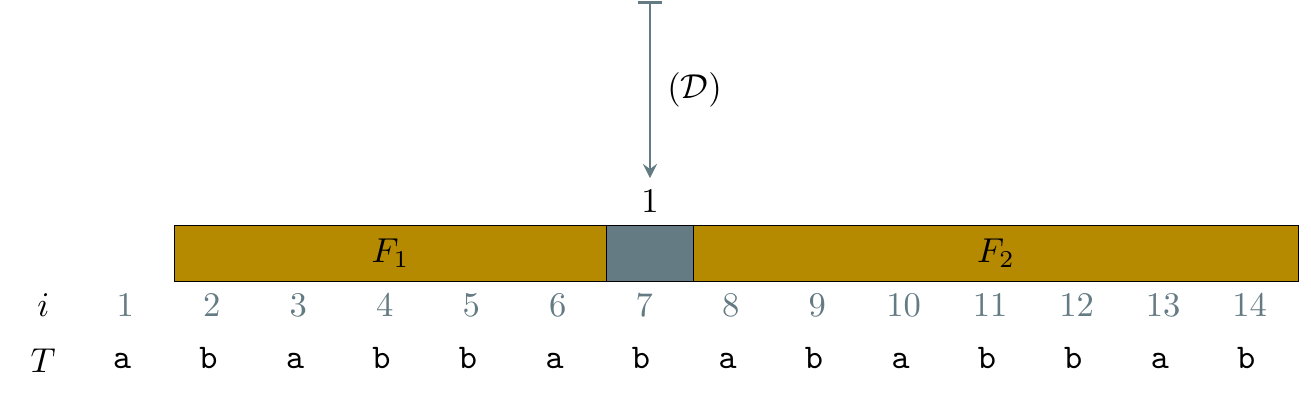}
        \caption{The referencing factor $F_2$ is introduced and Rule~\ref{itPLCPdecrease} is applied to the peak at $i=7$.}
        \label{figPLCPAlgoExample4}
    \end{subfigure}
    \caption{%
      Execution of our algorithm of \cref{secPLCPcompAlgo} computing the \protect\iPlcpcomp{} compression scheme on $T = \texttt{ababbabababbabbaababa\$}$.
        Due to limited space, we only illustrate the processing of the prefix $T[1 \twodots 14]$ in three steps (explained in \cref{exPlcpAlgo}).
        The vertical bars represent the \PLCP{} array, with the corresponding values written above, in text order from left ($i=1$) to right ($i=14$).
The shaded vertical bars represent the (current) \PLCP{} value of an interesting peak. Horizontal bars represent (referencing) factors.
In (b), the factor $F_1$, starting at position~$2$, is displayed as the maximal peak being \emph{tipped over} to the right.
    }
    \label{figPLCPAlgoExample}
\end{figure}

Since Rule~\ref{itPLCPdecrease} decreases at most the values of $\PLCP[\dst-\PLCP[\dst] \twodots \dst-1]$,
the highest peak~$\dst'$ in $\PLCP[1 \twodots \dst-1]$ is an interesting peak that is either
\begin{itemize}
  
    \item in the interval $[\dst-\PLCP[\dst] \twodots \dst-1]$, or,
    \item  in the case that all interesting peaks in $[\dst-\PLCP[\dst]\twodots{}\dst-1]$ are no longer interesting after decreasing their $\PLCP$ values,
       the rightmost peak preceding $\dst-\PLCP[\dst]$
       (whose \PLCP{} value is equal to the \PLCP{} value of the last peak removed from $L$ in Step~\ref{itAlgoApplyRules}).
\end{itemize}

\goodbreak

\noindent We can locate~$\dst'$ while applying Rule~\ref{itPLCPdecrease} as a result of creating the factor starting at $\dst$.
After locating~$\dst'$, we apply the following steps recursively:
\begin{enumerate}
    \item Substitute $T[\dst' \twodots \PLCP[\dst']-1]$ with a reference, because it is the highest peak in $T[1 \twodots \dst-1]$.
    \item If $\dst'' := \dst'+\PLCP[\dst']$ with $\PLCP[\dst'']$ $\ge$ $\threshold$ was not a peak, then~$\dst''$ becomes an interesting peak.
        In this case, substitute $\dst'$ with $\dst''$ in $L$ to preserve the order in~$L$.
        Otherwise, remove~$\dst'$ from~$L$.
    \item Split~$L$ into two sub-lists:
        \begin{itemize} 
            \item one containing text positions of the range $[1 \twodots \dst'-1]$, and
            \item the other containing text positions of the range $[\dst'+\PLCP[\dst'] \twodots \dst-1]$.
        \end{itemize}
    \item Recurse on each of the two sub-lists, i.e., find the highest peak in each sub-list and substitute it.
\end{enumerate}
This recursion is more efficient than the while-loop described in Lines~\ref{lineAlgStartProcessL} to~\ref{lineAlgEndProcessL} of \cref{algoRoutineAPrime}.

\begin{lemma}\label{lemValidPlcpcomp}
  The algorithm emits a valid \iPlcpcomp{} factorization of $T[1 \twodots \dst+\PLCP[\dst]-1]$.
\end{lemma}
After factorizing $T[1 \twodots \dst+\PLCP[\dst]-1]$,
we proceed with \cref{algoRoutineAPrime} on the remaining text $T[\dst+\PLCP[\dst] \twodots]$ to 
compute the factorization of the entire text.
It is left to explain how this algorithm can be adapted to the EM model efficiently.

\subsection{Factorization in External Memory}\label{secEMPLCPcomp}
Having the text, \PLCP{}, and $\Phi$ stored as files in EM,
we can compute the \iPlcpcomp{} scheme in three sequential scans over $n$ tuples and one sort operation:
\begin{enumerate}
\item Proceed with \cref{algoRoutineAPrime} to find pairs $(\dst, \ell = \PLCP[\dst])$ representing referencing factors $T[\dst\twodots\dst+\ell]$ by
scanning \PLCP{}.
\item Sort these pairs in ascending order of their $\dst$ components (i.e., in text order).
\item Simultaneously scan this sorted list of pairs and $\Phi$ to compute triplets of the form $(\dst, \src = \Phi[\dst], \ell)$, where the second component is the referred position of the referencing factor~$T[\dst\twodots\dst+\ell-1]$.
\item Finally, scan simultaneously the list of references and $T$ to replace each substring~$T[\dst \twodots \dst+\ell-1]$ by the reference $(\src, \ell)$
on reading the triplet $(\dst, \src, \ell)$.
\end{enumerate}

The pairs emitted during the \PLCP{} scan (Step~1) can be stored and then sorted in EM\@.
The references computed by the second scan can be written to disk for the final scan, which computes the \iPlcpcomp{} scheme of $T$ sequentially.
By doing so, no random access is required on the list of references.

During the \PLCP{} scan, the list $L$ can also be maintained on disk efficiently: until a maximal peak is found, we only append peaks to $L$.

Once a maximal peak $\dst$ has been found and a reference $(\dst, \ell)$ is emitted,
we scan over $L$ sequentially (a) to apply Rules~\ref{itPLCPdecrease} and~\ref{itPLCPreplace} and
(b) to find a remaining maximal peak, if any, in the process.
We then repeat this process until there are no more maximal peaks in $L$.
In practice, we scan the last elements of $L$ linearly from right to left, since only the last interesting peaks need to be updated.
For our experiments, we store~$L$ in RAM, as the number of elements was much lower than the following upper bound:

\begin{lemma}\label{lemmaListBounds}
   $\abs{L} = \Oh{\min(\sqrt{n \lg n}, r)}$, where $r$ is the number of BWT runs.
\end{lemma}
\begin{proof}
The list $L$ stores all interesting peaks between two different maximal peaks (or between the first position and the first maximal peak).
Given an interesting peak~$\dst$ with $\PLCP[\dst]$, there is no peak~$j$ with $\PLCP[j] \ge \PLCP[\dst]$ and $j < \dst < j + \PLCP[j]$.
In order to be added to $L$, the peak~$\dst$ must not be a maximal peak, i.e., there must be a
text position~$j$ with $\dst < j < \dst+\PLCP[\dst]$ and $\PLCP[j] > \PLCP[\dst]$.
The worst case is that $j = \dst+1$, $\PLCP[j] = \PLCP[\dst]+1$, and $j$ is again an interesting peak that is not maximal.
By induction, we may insert $m$ interesting non-maximal peaks $\menge{j_i}_{1 \le i \le m}$ into $L$ with $j_{i}+1 \le j_{i+1}$ for $1 \le i \le m-1$ and $\PLCP[j_i] \ge i$ for $1 \le i \le m$.

However, $\sum_{i=1}^m i \le \sum_{i=1}^m \PLCP[j_i] = \Oh{n \lg n}$ due to~\cite[Thm.\ 12]{karkkainen16irreducible}, such that $m = \Oh{\sqrt{n \lg n}}$.
From the same reference~\cite[Sect.\ 4]{karkkainen16irreducible}, we obtain that $m = \Oh{r}$.
\end{proof}

\begin{lemma}
There are texts of length~$n$ for which $\abs{L} = \Ot{\sqrt{n}}$.
\end{lemma}
\begin{proof}
For the proof, we use the following definition:
Given an interval~$I$, we define $\ibeg{I}$ and $\iend{I}$ to be the starting and the ending position of $I = [\ibeg{I}\twodots\iend{I}]$, respectively.

Let $\Sigma := \menge{\sigma_1,\ldots,\sigma_m}$ be an alphabet with $\sigma_1 < \sigma_2 < \ldots < \sigma_m$.
Set $F_m := \sigma_m$, and $F_i := \sigma_i F_{i+1} \sigma_i$ for $1 \le i \le m-1$.
Then our algorithm fills $L$ with \Ot{\sqrt{n}} interesting peaks on processing the text $T := F_m \cdots F_1$.
This is due to the following:

First, $\Phi[\ibeg{F_i}] = \ibeg{F_{i+1}}+1$ for each $i$ with $1 \le i \le m-1$, since
 \begin{itemize}
\item $T[\ibeg{F_i}\twodots] = F_i F_{i-1} \cdots = F_i \sigma_{i-1} F_i \sigma_{i-1} \cdots$ and
   \item $T[\ibeg{F_{i-j}}+j\twodots] = F_i \sigma_{i-1} \cdots \sigma_{i-j}$ for all $j$ with $0 \le j \le i-1$
 \end{itemize}
Hence,
$
T[\ibeg{F_{i-j}}+j\twodots]
\prec T[\ibeg{F_{i-1}}+1\twodots]
= F_i \sigma_{i-1} F_{i-2} \cdots
= F_i \sigma_{i-1} \sigma_{i-2} F_{i-1} \sigma_{i-2} \cdots
\prec T[\ibeg{F_i}\twodots]
$
for all $j$ with $2 \le j \le i-1$.
For all positions $1 \le j \le n$, we have $\lcp(T[j\twodots], T[\ibeg{F_i}\twodots]) \le \lcp(T[\ibeg{F_{i}}\twodots], T[\ibeg{T_{i-1}}+1\twodots]) = \abs{F_i}+1 = 2i$.
Hence, $\PLCP[\ibeg{F_i}] = 2i$ for each $i$ with $1 \le i \le m-1$.
Similarly, we obtain $\PLCP[\ibeg{F_i}+j] = 2i-j$  for each $j$ with $0 \le j \le \abs{F_i}$ and
$\PLCP[\iend{F_i}] = 2$ for each $i$ with $1 \le i \le m-1$.
We conclude that the text positions~$\ibeg{F_i}$ are interesting peaks, for $1 \le i \le m-1$.
Moreover, $\ibeg{F_{m-1}}$ is a maximum peak, since
$T[\ibeg{F_m}] = \sigma_1$ occurs only at $T[\ibeg{F_m}]$ and at the last text position $\iend{F_1}$ such that $\PLCP[\ibeg{F_m}] = 1$.

Finally, the algorithm collects $m-2$ interesting peaks before finding the maximal peak at text position $\ibeg{F_{m-1}}$.
Since $\abs{F_i} = 2i-1$, we have $\sum_{i=1}^m \abs{F_i} = \sum_{i=1}^m (2i-1) = n$, which holds for $m = \Ot{\sqrt{n}}$.
\end{proof}

 \section{Decompression}\label{sec:decompression}
The task of decompressing a bidirectional scheme is to \intWort{resolve} each reference~$(\src_i,\ell_i)$ of a referencing factor $T[\dst_i \twodots \dst_i+\ell_i-1$], i.e., to copy the characters from $T[\src_i \twodots \src_i+\ell_i-1]$ to $T[\dst_i \twodots \dst_i+\ell_i-1]$.

A unidirectional scheme can be decompressed by scanning linearly over the compressed input from left to right.
In that scenario, references can be resolved easily because they always refer to already decompressed parts of the text \cite{belazzougui16decoding}.
This property does not hold for a bidirectional scheme in general, as a reference can refer to a part of the text that again corresponds to a reference.
\vspace{.5em}
\begin{definition}[Dependency Graph]\label{def:depgraph}
Given a bidirectional factorization $F_1 \cdots F_b = T$, 
we model its references as a directed graph~$G$ with $V = \{v_1, \ldots, v_b\}$ such that there is a 1-to-1 correlation between nodes~$v_i$ and factors~$F_i$.
We add a directed edge $(v_i, v_j)$ from~$v_i$ to~$v_j$ with $i \ne j$ iff $F_i$ refers to at least one character in the factor~$F_j$.%
We put these edges into a set~$E$ to form a graph $G := (V,E)$ that has only literal factors as sinks.
A node $v_i$ can have more than one out-going edge if the referred substring is covered by multiple factors; in this case, we say $v_i$ is \intWort{multi-dependent} and call the set of its out-going edges a \intWort{multi-dependency}. The dependency graph of our example from \cref{figPLCPfactorization} can be seen in \cref{fig:depgraph}.
\end{definition}

\begin{figure}
    \vspace{.5em}
    \centering{%
        \begin{adjustbox}{valign=c}
            \includegraphics[width=0.3\textwidth]{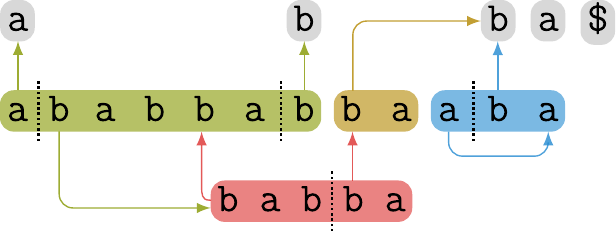}
        \end{adjustbox}
        \hspace{2cm}
        \begin{adjustbox}{valign=c}
            \includegraphics[width=0.46\textwidth]{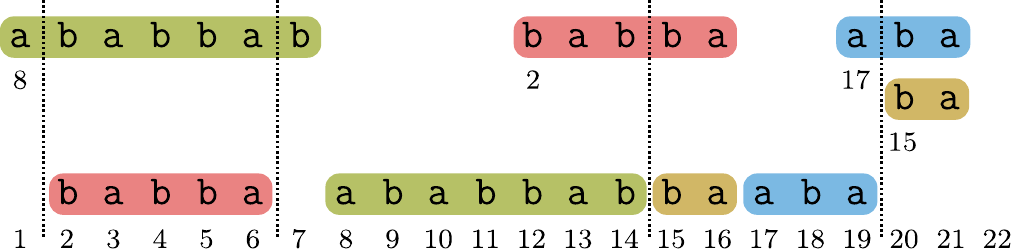}
        \end{adjustbox}
    }
    \caption{%
      The dependency graph (\emph{left}) and its EM representation  (\emph{right}) of the factorization given in \cref{figPLCPfactorization}.
      The multi-dependent factors of length seven and five have a cyclic dependency.
      The EM representation of the graph described in \cref{sec:decompression} consists of two copies of the list of all referencing factors, sorted by their source position (\emph{top}) as well as sorted by their destination (\emph{bottom}).
    }
    
    \label{fig:depgraph}
\end{figure}

Bidirectional decompressors face two challenges arising from this graph structure:
\begin{itemize}
     \item[\CustomLabel{itC1}{(C1)}]
     Long dependency chains (i.e., large values of $d(G)$) may affect the time and space complexity of decompression algorithms.
     \item[\CustomLabel{itC2}{(C2)}]
    The existence of multi-dependent nodes disallows efficient tree-based approaches.
\end{itemize}

In the remaining of this section, we present three strategies of attacking these issues, first individually (\cref{subsec:decomp-scan} and~\ref{subsec:compact}), and then together (\cref{subsec:pj}).
We focus on the resolution of indirect dependencies to obtain a dependency graph in which all references are direct children of literal factors.
After such a resolution, the text can be trivially recovered with $\sort(n)$~I/Os.

\subsection{Decompressor \texorpdfstring{\protect\strScan{}}{scan}}\label{subsec:decomp-scan}
The decompressor \strScan{} was introduced in~\cite[Sect.~3.2.2]{dinklage17tudocomp} (to which we refer for a detailed description).
In its main phase, \strScan{} avoids multi-dependencies by splitting each reference~$(\src,\ell)$ with $\ell > 1$ into references $(\src,1),\ldots,(\src+\ell-1,1)$, i.e., one for each character.
Then any undecoded position refers to either a literal factor or another reference.
Hence the underlying dependency graph becomes a forest,
which can conceptionally be resolved in \Oh{n}~time using standard traversal techniques.
The initial splitting may however increase the number of references by a factor of $\Oh{n}$ causing inefficiencies and a significant memory overhead (which \strScan{} tries to reduce heuristically by preprocessing).
This strategy is also similar to the parallel LZ77 decompressor of \citet[Sect.\ 4.2]{farach95parallel}.

\subsection{Optimizing the Coding for Decompression}\label{subsec:compact}
Orthogonally, we present the novel approach \strIMCompact{} to improve an existing bidirectional coding for decompression by shortening dependency chains (see the left sub-figure of \cref{fig:compaction}).
This approach neither changes the factorization nor does it convert a referencing factor into a literal.
It may be used directly after the compression step to accelerate future decompression.

Given a coding, we construct its dependency graph $G$ but omit all multi-dependencies. 
As a result, we obtain a forest in which each reference depends only on a unique predecessor as illustrated in \cref{fig:compaction} (middle).
Using a top-down traversal (e.g., BFS) on each tree individually, we can replace all chains with direct references to the root.
Building~$G$ and traversing it requires \Oh{b}~total time.

\begin{figure}
    \centering{%
        \includegraphics[width=\linewidth]{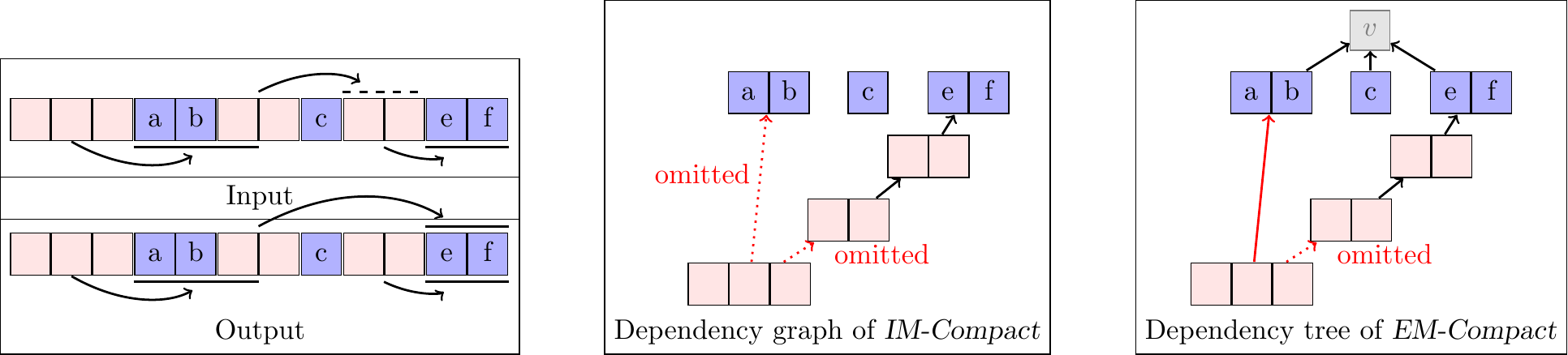}
      }
    \caption{%
        Compaction of a bidirectional scheme.
        \emph{Left}: The factors of the input are represented by maximal consecutive blocks of the same shading.
        In this example, the input consists of six factors.
        Referencing factors store no characters, have a light shading and an out-going arrow pointing to a vertical bar representing its corresponding reference.
        The first factor refers to two factors and is not resolved during compaction.
        The third factor refers to the fifth which refers to the sixth;
        this chain is compacted by redirecting the third factor to the sixth directly.
        \emph{Middle} and \emph{Right}:
        Dotted edges indicate dependencies with no corresponding edge in the algorithms described in \cref{subsec:compact}.
    }

    \label{fig:compaction}
\end{figure}

\noindent We now present \strEMCompact{}, an I/O-optimal variant of \strIMCompact{}:

\def\vecFac{\ensuremath{\mathsf{factors}}}
\def\pqSplit{\ensuremath{\mathsf{PQSplit}}}
\def\vecChld{\ensuremath{\mathsf{requests}}}
\def\vecNChld{\ensuremath{\mathsf{nextRequests}}}
\def\vecPQ{\ensuremath{\mathsf{PQ}}}
\def\vecRes{\ensuremath{\mathsf{result}}}

\begin{enumerate}[label={\textbf{Step~\arabic*}:},ref={Step~\arabic*},wide, labelwidth=!, labelindent=0pt] 
\item\label{it:construct-vectors} We first construct a representation of the dependency graph consisting of two EM vectors \vecChld{} and \vecFac{}.
Intuitively, each reference (child) sends a request message to the first factor it refers to (parent).
Addressing is implemented indirectly in terms of text positions rather than factor indices.
To this end,
for each reference $(\src, \ell)$ corresponding to a factor~$F_i = T[\dst\twodots\dst+\ell-1]$,
we push
(i) the tuple $\jbracket{\src, \ell, i}$ into \vecChld{}, and
(ii) the tuple $\jbracket{\dst, \ell, i}$ into \vecFac{}.
Additionally, each literal factor~$F_i = T[\dst\twodots\dst+\ell-1]$ contributes a tuple~$\jbracket{\dst, \ell, i}$ to \vecFac{}.
Subsequently, we sort\footnote{%
    To sort tuples we always use lexicographic order, i.e., we order tuples as implied by the first unequal element.}
both vectors independently, bringing the messages in \vecChld{} and the recipients in \vecFac{} into the same order.

\item\label{it:first-tree-pass} We now scan through \vecFac{} and \vecChld{} simultaneously.
  By doing so, each $F_i$ in \vecFac{} can gather all its children (requests):
a factor $F_i$ with tuple~\jbracket{\dst,\ell,i} has a child~\jbracket{\src,\ell',i'} if $\src \in [\dst, \dst+\ell)$.
The factor of this child~$F_{i'}$ is completely contained in~$F_i$ if $\src+\ell' \le \dst+\ell$.
Otherwise, $F_{i'}$ is multi-dependent.
In contrast to \strIMCompact{}, which discards such a multi-dependency completely,
\strEMCompact{} retains one edge to obtain a connected dependency tree simplifying \ref{it:list-ranking}.\footnote{%
    \strEMCompact{} keeps multi-dependent nodes despite its inability to optimize them. It does so because a subtree rooted in a multi-dependent node can contain optimizable dependency chains.}
To complete the tree, we add a virtual node $v$ and assign all literal factors as $v$'s children.
The resulting graph is a tree rooted in $v$ with $b{+}1$ nodes as illustrated in \cref{fig:compaction} (right).
Its construction requires $\sort(b)$~I/Os.%
\footnote{Depending on the encoding of the input, a scan over the content of literal factors may be necessary and trigger $\scan(n)$~I/Os.}

\item\label{it:list-ranking} Subsequently, we apply the Euler tour technique and list ranking~\cite[Sect.~3.6]{DBLP:conf/dagstuhl/MaheshwariZ02} on the tree built in~\ref{it:first-tree-pass} to calculate the depths of all nodes, triggering $\sort(b)$~I/Os.

\item\label{it:compact-resolution} With an additional tree traversal, we can finally update the referred positions.
For that, we annotate each tuple in \vecFac{} and \vecChld{} with the depth of its corresponding node.
Then the vectors are sorted by the depths of their items and, in case of equality, the order used in~\ref{it:construct-vectors}.
Similarly to~\ref{it:first-tree-pass}, we scan both vectors simultaneously to traverse the dependency tree.

Due to the order of both vectors, \strEMCompact{} processes nodes layer-wise and within each layer from left-to-right.
Thus, parents are processed before their children, and can inductively forward their referred-to positions to their children.
Following the time-forward processing~\cite[Sect.~3.4]{DBLP:conf/dagstuhl/MaheshwariZ02} technique, we transport those updates as messages in an EM priority queue~\vecPQ{}.

When processing node $F_i$ at depth $d$, we check whether a message of the form $\jbracket{d, i, \src_1, \dst_1}$ is at the top of \vecPQ{}.
If so, we dequeue it and update the referred position of~$F_i$ to $\src_1+\src_0-\dst_1$, where $\src_0$ is the former referred position of~$F_i$ as illustrated in \cref{figOneSplit}.
In any case, we iterate over all non multi-dependent children:
for each $F_j$, we push the message \jbracket{d{+}1, j, \src_1, \dst_1} into~\vecPQ{}.
\end{enumerate}

During each step, $\Oh{b}$ items are sorted and scanned, triggering $\sort(b)$~I/Os in total.
I/O-optimality follows by a reduction to the permutation problem analogously to the construction in~\cite[Thm. 1]{belazzougui16decoding}.

\begin{figure}
    \centering{%
        \includegraphics[scale=1.0]{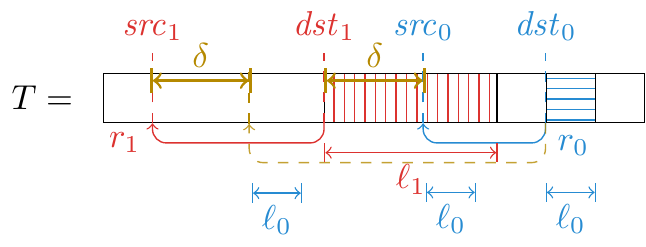}
        \hfill
        \includegraphics[scale=1.0]{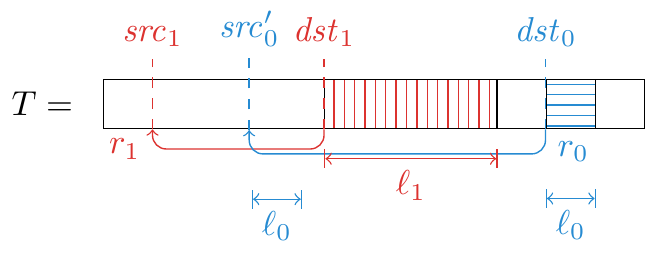}
    }%
    \caption{Pointer jumping of the reference $r_0 := (\src_0,\ell_0)$ belonging to the factor starting at $\dst_0$.
        We set the referred position of this reference to $\src_0' := \src_1+\delta$, where $\delta = \src_0-\dst_1$ and $T[\dst_1 \twodots \dst_1+\ell_1-1]$ is a factor with $\dst_1 \le \src_0 \le \src_0+\ell_0-1 \le \dst_1+\ell_1-1$ having the reference $r_1$ to text position $\src_1$.
       The left and the right picture show the setting before and after applying the pointer jumping, respectively.
         }
    \label{figOneSplit}
\end{figure}

\subsection{Decompressor \texorpdfstring{\protect\strPJ{}}{EM-PJ}}\label{subsec:pj}
Our novel decompressor \strPJ{} (refer to \cref{decompPJ} for details) adapts the ideas of the coding optimizers~\strIMCompact{} and~\strEMCompact{} for decompression.
While \strEMCompact{} is I/O-optimal, its resolution phase (\ref{it:compact-resolution}) relies on the fact that we can efficiently find a topological order of the dependency \emph{tree}.
Unfortunately, this is not the case for general DAGs induced by factorizations with multi-dependencies.

We switch to the pointer jumping technique~\cite[Sect. 2.2]{DBLP:books/aw/JaJa92} for dependency resolution.

Let $G$ be the dependency graph of the factorization $T = F_1 \cdots F_b$. 
As a starter, we assume that all factors are single-dependent, 
i.e., each node $v$ representing  a referencing factor has exactly one outgoing edge $(v, p(v))$. 
For all other nodes (representing literal factors) we define $p(v) := v$. 
Clearly, like in \strEMCompact{}, $G$ forms a forest in which each tree is rooted in a literal factor. 
When applying the pointer jumping technique, we take each referencing factor and attach it to the parent of its parent 
(cf.\ \cref{figPointerJumping}). 
Given that $G'$ is the resulting graph with $p'(v) = p(p(v))$,
we thereby halve the depth, i.e., $d(G') = \upgauss{d(G)/2}$ if $d(G) \ge 2$, where $d(G)$ denotes the maximum depth of a tree in $G$. 
Hence, after $\Ot{\lg d(G)}$ iterations all indirect references are resolved and have been replaced by direct references to literal factors.

If we allow multi-dependencies, pointer jumping is only possible for single-dependent nodes. 
To apply pointer jumping, we split each multi-dependent reference into the smallest possible set of single-dependent references.
A split is introduced ad-hoc each time it is required for a pointer jump.
The details of the splitting are discussed in \cref{decompPJ}.

 \begin{figure}
    \centering{%
        \begin{adjustbox}{valign=c}
        \includegraphics[width=0.45\textwidth]{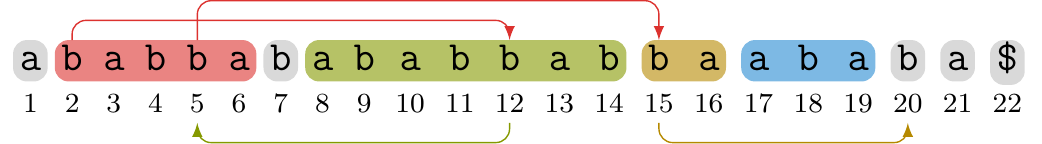}
        \end{adjustbox}
        \hfill
        \begin{adjustbox}{valign=c}
        \includegraphics[width=0.45\textwidth]{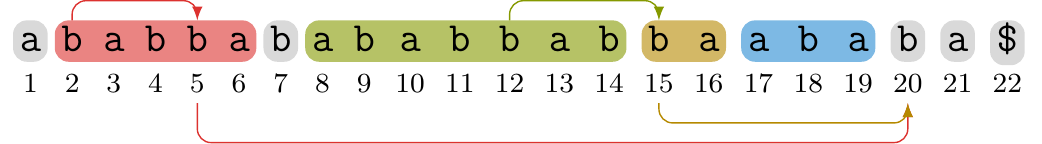}
        \end{adjustbox}
}%

\vspace{1em}
\hspace{4em}
        \begin{adjustbox}{valign=c}
\rotatebox{-90}{%
        \begin{forest}
            for tree={draw=gray}
            [\rotatebox{90}{20},fill=gray!30,
            [\rotatebox{90}{15},fill=solarizedYellow!60,
            [\rotatebox{90}{5},fill=solarizedRed!60,
            [\rotatebox{90}{12},fill=solarizedGreen!60,
            [\rotatebox{90}{2},fill=solarizedRed!60,
            ]]]]]
        \end{forest}
}%
        \end{adjustbox}
\hfill
        \begin{adjustbox}{valign=c}
\rotatebox{-90}{%
        \begin{forest}
            for tree={draw=gray}
            [\rotatebox{90}{20},fill=gray!30,
            [\rotatebox{90}{15},fill=solarizedYellow!60,
            [\rotatebox{90}{12},fill=solarizedGreen!60,]]
            [\rotatebox{90}{5},fill=solarizedRed!60,
            [\rotatebox{90}{2},fill=solarizedRed!60,
            ]]]
        \end{forest}
}%
        \end{adjustbox}
\hspace{4em}
    \caption{Pointer jumping applied to references. Suppose that our example text is represented by the coding described in \cref{figPLCPfactorization}.
        To extract the character $T[2]$, we need to resolve the reference~$(12,5)$, which has a depth of three (\emph{bottom left} figure).
        In case that we split all references into references of length one,
        we can reduce the depth of the reference associated with $T[2]$ by pointer jumping (\emph{right} figure).
        The order in which this technique is applied to the references has an impact on the resulting references.
        Here, we assumed that we can apply this technique \emph{in parallel}.
    }
    \label{figPointerJumping}
 \end{figure}

\begin{figure}
    \centering{%
        \begin{adjustbox}{valign=t}
            \includegraphics[width=0.45\textwidth]{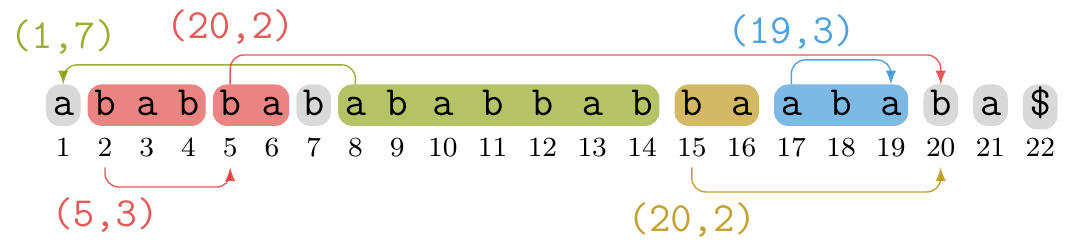}
        \end{adjustbox}
        \hfill
        \begin{adjustbox}{valign=t}
            \includegraphics[width=0.45\textwidth]{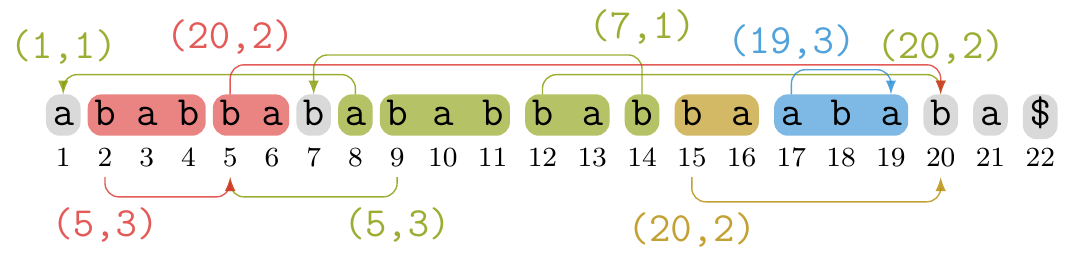}
        \end{adjustbox}
    }%
    \caption{Split-Strategy of \protect\strPJ{} applied to the first (\emph{left} figure) and second (\emph{right} figure) referencing factor of the factorization given in \cref{figPLCPfactorization}.
      \protect\strPJ{} splits up references in a minimal number of sub-references on which the pointer jumping technique can be applied.
        The \emph{left} figure shows such an application to the reference of the leftmost referencing factor that is split into two sub-references.
        The first and second sub-reference receive new referred positions based on the referred positions of the second and third referencing factors, respectively.
        In the \emph{right} figure, we split up the next reference~$(1,7)$ in four sub-references, where the first and last sub-reference refer to literal factors.
    }
    \label{figStrategyPJ}
\end{figure}

Like in \strEMCompact{}, we construct a representation of the dependency graph consisting of two EM vectors called \vecChld{} and \vecFac{}.
Intuitively, each request (child) sends a request message to the first factor it refers to (parent).
Addressing is implemented indirectly in terms of text positions rather than factor indices.
For each reference $(\src, \ell)$ corresponding to a factor~$F_i = T[\dst\twodots\dst+\ell-1]$, we push $\jbracket{\dst, \ell, \src}$ into \vecChld{} and $\jbracket{\src, \ell, \dst}$ into \vecFac{}.
We omit literal factors, since the lack of a reference in \vecFac{} for a certain text position indicates the presence of a literal factor.

Subsequently, we sort both vectors independently, bringing the messages in \vecChld{} and the recipients in \vecFac{} into the same order. On the right side of \cref{fig:depgraph} we see a visualization of the lists (after the initial sorting) for our running example.
We augment~\vecChld{} with an initially empty EM priority queue~\pqSplit.
In the following, after processing a factor~$F_i$, 
we write $F_i$ either to a vector~\vecRes{} if it refers to literal factors, or to a vector~$\vecNChld$ otherwise:
Let $\jbracket{\dst, \ell, \src}$ be the smallest unprocessed request of a factor~$F_i$ received via \vecChld{} or \pqSplit{}.
If it originates from \vecChld{}, we advance \vecChld{}'s read pointer for the next iteration, 
otherwise we dequeue the top element from~\pqSplit{}.
We process the read request $\jbracket{\dst, \ell, \src}$ depending on the following cases (cf.\ \cref{figStrategyPJ}):

\begin{description}
    \item[Jump] The request is completely covered by parent $F_j$ in \vecFac{}.
      In this case, we substitute $F_i$'s reference according to $F_j$ and push it into $\vecNChld$ to be processed in the next iteration.

    \item[Finalize] No parent (partially) overlapping with $F_i$ is available in \vecFac{}.
    Then we know that $F_i$ points to a substring contained in literal factors. We finalize $F_i$ by pushing it into \vecRes{}.

  \item[Split] A prefix of $F_i$ is contained in the parent $F_j$ or points to literals. Let $\ell' < \ell$ be the length of the longest such prefix.
    Then split $F_i$ into a prefix $F^{\textup{P}}_i$ of length $\ell'$ and a suffix $F^{\textup{S}}_i$ of length $\ell - \ell'$.
    By construction, either case `Jump' or case `Finalize' is applicable to $F^{\textup{P}}_i$, and we execute it directly.
    Then we push $\jbracket{\src{+}\ell', \ell{-}\ell', \dst+{\ell'}}$ representing $F^{\textup{S}}_i$ into \pqSplit{} to process it later within the same iteration.
    Observe that $F_i$ can be split multiple times during the same iteration.
\end{description}

\noindent If \vecNChld{} is not empty, we sort it and recurse by processing \vecNChld{} and the (unaltered) \vecFac{} simultaneously as before.
With these steps, we obtain the final result:%

\begin{theorem}\label{lemmaPJ}
    Let $F_1 \cdots F_b = T$ be a $\threshold$-restricted bidirectional scheme, and $d(G) < b$ be the depth of $T$'s dependency graph~$G$.
    Then \strPJ{} requires $\Oh{\lg\left(d(G)\right) \sort(n / \threshold)}$~I/Os.
\end{theorem}
\begin{proof}
    As pointer jumping halves the depth of the dependency graph $G$, \strPJ{} performs $\Oh{\lg d(G)}$ iterations.
    While $G$ changes in each step, it remains valid in terms of \cref{def:depgraph}.
    Despite \strPJ{} introducing new nodes by splitting a factor, the number $\abs{V}$ of nodes of~$G$ is bounded by the maximal number $k := \Ot{n / \threshold}$ of factors, i.e., $\abs{V} = \Oh{k}$.

    Hence the size of the vectors involved is bounded as follows:
    \vecFac{} is filled once where each factor contributes at most one element, thus $|\vecFac| = \Oh{b} = \Oh{k}$.
    \vecChld{} is reproduced in each iteration and may reach up~to $|\vecChld| = \Oh{\abs{V}} = \Oh{k}$ items,
    which directly translates into an upper bound on the number of items in \pqSplit.
    Analogously, $|\vecRes{}|$ is bounded from above by $\Oh{k}$.
    Thus in each round $\Oh{k}$ items are sorted and scanned a constant number of times, resulting in the claimed bound.
\end{proof}

\subsection{Detailed Description of \texorpdfstring{\protect\strPJ{}}{EM-PJ}}\label{decompPJ}
For completeness, we present a more technical representation of our decompression strategy \strPJ{}.
As previously explained, this strategy is based on the optimization technique for improving the decompression of a coding described in \cref{subsec:compact}.
The difference is that \strPJ{} splits references up in a minimal number of sub-references, where each sub-reference either (a) can be immediately decoded or (b) has a referred position that can be set to the referred position of the reference it refers to.
Suppose that a reference~$r$ refers to a substring~$S$ that is not completely decompressed.
We split $r$ into sub-references such that
a sub-reference~$(\src_0,\ell_0)$ refers to a substring~$T[\src_0 \twodots \src_0+\ell_0-1]$ that is either
\begin{itemize}
    \item already decompressed, or
    \item contained in a substring $T[\dst_1 \twodots \dst_1+\ell_1-1]$ with $\dst_1 \le \src_0 \le \src_0+\ell_0-1 \le \dst_1+\ell_1-1$ substituted by a reference~$(\src_1,\ell_1)$, cf.\ \cref{figOneSplit}.
\end{itemize}
In the former case, we can resolve the sub-reference.
In the latter case, we exchange $(\src_0,\ell_0)$ with $(\src_1+\src_0-\dst_1,\ell_0)$.
Due to the pointer jumping technique, we need $\Oh{\lg d} = \Oh{\lg n}$ scans of the references to resolve all references,
where $d \le n$ is the maximal depth a reference can have.

\newcommand*{\ListReq}{\ensuremath{L_{\textup{req}}}}
\newcommand*{\ListRef}{\ensuremath{L_{\textup{ref}}}}
\newcommand*{\ListRes}{\ensuremath{L_{\textup{res}}}}
\newcommand*{\ListReqN}{\ensuremath{L_{\textup{req}}^{\text{new}}}}
\newcommand*{\ListRefN}{\ensuremath{L_{\textup{ref}}^{\text{new}}}}
\newcommand*{\ListResN}{\ensuremath{L_{\textup{res}}^{\text{new}}}}
\newcommand*{\ListResP}{\ensuremath{\hat{L}_{\textup{res}}}}

The strategy \strPJ{} maintains the following lists in external memory:
\begin{itemize}
    \item the list of requests \ListReq{} storing tuples $(\src, \dst, \ell)$ to maintain the information that we request the substring $T[\src \twodots \src+\ell-1]$ to restore $T[\dst \twodots \dst+\ell-1]$,
    \item the list of references \ListRef{} storing tuples $(\dst, \src, \ell)$ corresponding to unresolved referencing factors for applying the pointer jumping technique, and
    \item the list of resolutions \ListRes{} storing tuples $(\dst, S)$ with $S \in \Sigma^*$ for the instruction to copy the string $S$ to $T[\dst \twodots \dst+\abs{S}-1]$.
\end{itemize}

\subparagraph*{Initial Step}
   We create an external file~$T$ with the length of the original text and scan sequentially the list of factors represented by their coding.
We process the $x$-th factor~$F_x$ as follows:
\begin{itemize}
    \item If $F_x$ is a literal factor, copy its contents to $T[1+\abs{F_1\cdots F_{x-1}}]$.
    \item Otherwise, $F_x$ is a referencing factor. Given its reference is~$(\src,\ell)$,
    store $(\dst,\src,\ell)$ in $\ListRef$, and $(\src,\dst,\ell)$ in $\ListReq$, where $\dst$ is the starting position of $F_x$.
\end{itemize}
Subsequently, sort the request tuples by their first component (the source position).
The reference list~$\ListRef$ is already sorted with respect to the first component of its tuples.

\subparagraph*{Recursion Step}
After the initial step, we process the three lists until every reference got resolved.
For that, we treat the three lists \ListReq{}, \ListRef{}, and \ListRes{} as queues, discarding a read tuple as it will no longer be needed.
Additionally to these lists, we create new lists \ListReqN{}, \ListRefN{}, and \ListResN{} whose contents we fill during a scan.
During a scan, we process the three lists \ListReq{}, \ListRef{}, and \ListRes{} simultaneously with respect to their first components.
Whenever we are at a text position that is equal to the first component of a resolution of \ListRes{} or request of \ListReq{} we take action:

First, suppose that we are at a position~$\dst$ and that there is a resolution $(\dst, S)$ at the head of $\ListRes$.
In this case, we set $T[\dst \twodots \dst+\abs{S}-1] \gets S$.

\begin{figure}
    \centering{%
        \begin{adjustbox}{valign=t}
            \includegraphics{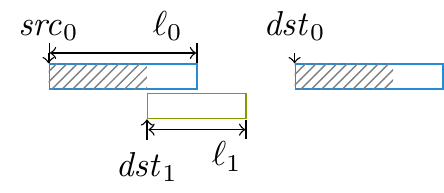}
        \end{adjustbox}
        \hfill
        \begin{adjustbox}{valign=t}
            \includegraphics{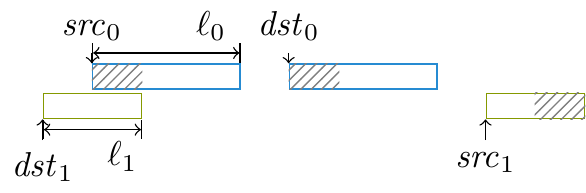}
        \end{adjustbox}
    }%

    \centering{%
            (1) $\src_0 \le \dst_1 - 1$
            \hfill
            (2) $\src_0 \ge \dst_1$
    }%
    \caption{Cases studied in \cref{decompPJ}. In the left figure (Case~1), the shaded part $T[\src_0 \twodots \src_0+\ell_0-\dst_1]$ is decoded and can be copied to
        $T[\dst_0 \twodots \dst_0+\ell_0-\dst_1]$.
        In the right figure (Case~2),
        the shaded part $T[\src_0 \twodots \dst_1+\ell_1-1]$ was substituted by the reference $(\dst_1,\src_1,\ell)$.
        Following this reference, the substring $T[\src_0 \twodots \dst_1+\ell_1-1]$ can be decoded by decoding
        $T[\src_1-\src_0-\dst_1 \twodots \src_1+\ell_1-1]$.
    }
    \label{figDecompPJ}
\end{figure}

Second, suppose that we are at a position~$\src_0$ and there is a request $(\src_0,\dst_0,\ell_0)$ at the head of $\ListReq$.
We scan the list of references~$\ListRef$ for a reference $(\dst_1, \src_1, \ell_1)$ with the smallest $\dst_1$ with $\src_0 \le \dst_1 + \ell_1 - 1$ while discarding all entries whose first component precedes~$\dst_1$.
If $\src_0 + \ell_0 \le \dst_1$, then $T[\src_0 \twodots \src_0+\ell_0-1]$ is already decoded; hence we can
insert $(\dst_0, T[\src_0 \twodots \src_0+\ell_0-1])$ into the new resolution list~$\ListResN$.
Otherwise ($\src_0 + \ell_0 > \dst_1$), we consider two cases (cf.~\cref{figDecompPJ}):

\begin{itemize}
    \item[Case 1:]  $\src_0 < \dst_1$.
    In this case, $T[\src_0 \twodots \dst_1-1]$ is already resolved, and we insert $(\dst_0, T[\src_0 \twodots \dst_1-1])$ into the new resolution list~$\ListResN$.
    We update the request $(\src_0, \dst_0, \ell_0)$ to $(\dst_1, \dst_0+\dst_1-\src_0, \ell_0-\dst_1+\src_0)$, and proceed with Case~2.

    \item[Case 2:]  $\src_0 \ge \dst_1$.
    If $\dst_1 + \ell_1 - \src_0 \ge \ell_0$, then we can pointer jump the request $(\src_0, \dst_0, \ell_0)$ to
    $(\src_1+\src_0-\dst_1, \dst_0, \ell_0)$.
    Otherwise, we split the request in two requests $(\src_0,\dst_0,\dst_1+\ell_1-\src_0)$ and $(\dst_1+\ell_1,\dst_0+\dst_1+\ell_1-\src_0,\ell_0-\dst_1-\ell_1+\src_0)$.
    We can pointer jump the first request to $(\src_1+\src_0-\dst_1, \dst_0, \dst_1+\ell_1-\src_0)$.
\end{itemize}
In both cases, when creating a new request~$(\src,\dst,\ell)$, we insert it into the new request list~$\ListReqN$, and insert $(\dst,\src,\ell)$ into the new reference list~$\ListRefN$.

After the scan, we move the contents of the new lists \ListReqN{}, \ListRefN{}, and \ListResN{} to their corresponding lists \ListReq{}, \ListRef{}, and \ListRes{}, respectively.
We repeat this process until the list of references~\ListRef{} becomes empty.

Due to practical issues, we did not implement the lists of resolutions~\ListRes{} and \ListResN{} as explained,
since an entry of these lists would hold a string of \emph{arbitrary} length.
Instead, we use a single list $\ListResP$ storing tuples $(\src,\dst,\ell)$ saying that the substring $T[\src \twodots \src+\ell-1]$ is already decoded and can be copied to $T[\dst \twodots \dst+\ell-1]$.
To avoid random I/O, we first process a request~$(\src,\dst,\ell) \in \ListReq$ completely with Case~1 and Case~2 before selecting the next request~$(\src',\dst',\ell') \in \ListReq$. However, both requests can overlap with $\src \le \src' \le \src+\ell$ such that $\ListResP$ can become unsorted (cf.~\cref{figUnorderedResolutionList}).
We sort $\ListResP$ after a scan of all requests according to its first component.
Subsequently, we scan the text and \ListResP{} to produce, given a tuple $(\src,\dst,\ell) \in \ListResP$, the tuples $(\dst+i,T[\src+i])$ for all integers $i$ with $0 \le i \le \ell-1$.
These tuples are sorted by their first component.
With a linear scan on $T$, we set $T[\dst] \gets c$ for each such tuple~$(\dst,c)$.

\begin{figure}
    \centering{%
        \includegraphics[scale=1.0]{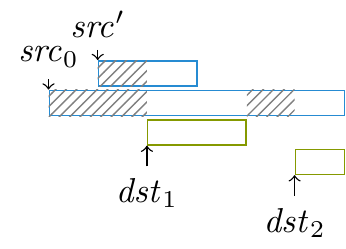}
    }%
    \caption{Unordered insertion into $\ListResP$. Suppose that the first tuples in the request list~$\ListReq$ are $(\src_0,\dst_0,\ell_0)$ and $(\src',\dst',\ell')$ with $\src_0 \le \src'$ and that the first tuples in the reference list~$\ListRef$ are $(\dst_1,\src_1,\ell_1)$ and $(\dst_2,\src_2,\ell_2)$ with
        $\src' < \dst_1 \le \src'+\ell'-1$.
        Since we first process $(\src_0,\dst_0,\ell_0)$ and its resulting sub-requests by Cases~1 and~2, we produce the resolution
        to copy $T[\dst_1+\ell_1 \twodots \dst_2-1]$ to $T[\dst_0+\dst_1+\ell_1-\src_0 \twodots \dst_0+\dst_2-1-\src_0]$ prior to producing the resolution
        to copy $T[\src' \twodots \dst_1-1]$ to $T[\dst' \twodots \dst'+\dst_1-1-\src']$.
    }
    \label{figUnorderedResolutionList}
\end{figure}

\section{Practical Evaluation}
We embedded our algorithms in the C++ framework tudocomp, available at \url{https://github.com/tudocomp/tudocomp}.
We collected the used EM algorithms for constructing the needed text data structures in the repository \url{https://github.com/tudocomp/emtools}.

\subparagraph*{Experimental Setup}
Our experiments are conducted on a machine with \SI{16}{\gibi\byte} of RAM\footnote{In order to avoid swapping, each experiment was conducted with a limit of \SI{14}{\gibi\byte} of RAM.},
eight \texttt{Hitachi HUA72302} hard drives with 1.8 TiB,
two \texttt{Samsung SSD 850} SSDs with 465.8 GiB, and
an \texttt{Intel Xeon CPU}~\texttt{i7-6800K}.
The operating system is a 64-bit version of Ubuntu Linux~\num{16.04}.
We implemented \iPlcpcomp{} in the version 1.4.99~(development snapshot) of the STXXL library~\cite{dementiev08stxxl}.
We compiled the source code with the GNU \texttt{g++~7.4} compiler with the compile flags \texttt{-O3 -march=native -DNDEBUG}.

\subparagraph*{Text Collections}
\begin{table}
	\centerline{%
		\begin{tabular}{r*{8}{r}}
			\toprule
			\multicolumn{4}{l}{\texttt{commoncrawl}} \\
			prefix length & $H_0$ & $H_1$ & $H_2$ & $H_3$ & $H_4$ & $H_5$ & $H_6$ & $H_7$
			\\\midrule
			16 GiB &
			5.99165 &
			4.26109 &
			3.48920 &
			2.94113 &
			2.42738 &
			2.01886 &
			1.64558 &
			1.35130
			\\
			32 GiB &
			5.99145 &
			4.26160 &
			3.49006 &
			2.94411 &
			2.43471 &
			2.03284 &
			1.66737 &
			1.37798
			\\
			64 GiB &
			5.99119 &
			4.26209 &
			3.49100 &
			2.94669 &
			2.44088 &
			2.04409 &
			1.68482 &
			1.40001
			\\
			128 GiB  &
			5.99177 &
			4.26148 &
			3.49055 &
			2.94684 &
			2.44231 &
			2.04753 &
			1.69087 &
			1.40839
			\\\bottomrule
			\toprule
			\multicolumn{4}{l}{\texttt{dna}} \\
			prefix length & $H_0$ & $H_1$ & $H_2$ & $H_3$ & $H_4$ & $H_5$ & $H_6$ & $H_7$
			\\\midrule
			16 GiB &
			1.9715 &
			1.94676&
			1.93166&
			1.92232&
			1.91167&
			1.89491&
			1.87101&
			1.84585
			\\
			32 GiB &
			1.97128&
			1.94561&
			1.93201&
			1.92421&
			1.91507&
			1.90190&
			1.88270&
			1.86160
			\\
			64 GiB &
			1.97067&
			1.94506&
			1.93145&
			1.92424&
			1.91588&
			1.90445&
			1.88763&
			1.86889
			\\
			128 GiB  &
			1.97528&
			1.95010&
			1.93873&
			1.93273&
			1.92486&
			1.91341&
			1.89601&
			1.87634
			\\\bottomrule
		\end{tabular}
	}%
	\caption{Empirical entropies of our data sets. 
		The alphabet sizes of all instances are 242 and 4 for \texttt{commoncrawl} and \texttt{dna}, respectively.}
	\label{tableEntropy}
\end{table}

We conduct our experiments on two texts of different alphabet sizes and repetitiveness (cf.\ \cref{tableEntropy}):
\begin{itemize}
    \item \textsc{commoncrawl}: A crawl of web pages with an alphabet size of 242 collected by the commoncrawl organization.
    \item \textsc{dna}: DNA sequences with an alphabet size of 4 extracted from FASTA files.
\end{itemize}

\subparagraph*{Algorithms}
We compare \iPlcpcomp{} against \iEMLPF{}~\cite{karkkainen14parsing} by \citeauthor{karkkainen14parsing}, 
which is an EM algorithm computing the \iLZSS{} factorization by constructing the \LPF{} array. 
In addition to the input text, it requires \SA{} and \LCP{}.

In early experiments with \iLZscan{}~\cite{karkkainen14parsing}, 
it became clear that its throughput on the text collection we use is nowhere near competitiveness with \iEMLPF{} and \iPlcpcomp{}. 
Therefore, it is not considered in our experiments. 
Semi-external \iLZSS{} algorithms like \iSEKKP{}~\cite{karkkainen14parsing} 
storing the text or parts of the text in RAM have not been considered.

\subparagraph*{Data Structures}
Currently, the fastest way to compute the data structures $\PLCP$ and $\Phi$ in EM is
to compute $\BWT$ from $\SA$ 
with the parallel EM algorithm \iEMBWT{} by K\"arkk\"ainen and Kempa\footnote{\url{https://www.cs.helsinki.fi/u/dkempa/pem_bwt.html}},
and use it for computing $\PLCP$ with 
the parallel EM construction algorithm of~\citet{karkkainen17plcp}. 
We modified the source code of the latter to also produce $\Phi$ as a side product.
This chain of algorithms is illustrated in \cref{figPLCPschedule}.

For \iEMLPF{}, we additionally need to convert \PLCP{} to \LCP{} by a scan over $\SA{}$
and a subsequent sort step.
This is currently the fastest approach for obtaining \LCP{}, 
as other approaches building \LCP{} directly from \SA{} like~\cite{karkkainen16lcpscan} are slower.

Consequently, both contestants need (directly or indirectly) \SA{}.
However, it takes a considerable amount of time to construct it 
with EM algorithms on a single machine (e.g., with pSAScan~\cite{karkkainen15psascan}).
To put the focus on the comparison between \iEMLPF{} and \iPlcpcomp{}, we do not take into account the construction of \SA{} and \LCP{} when measuring running times.

\begin{figure}
    \centering{%
        \includegraphics[scale=1.0]{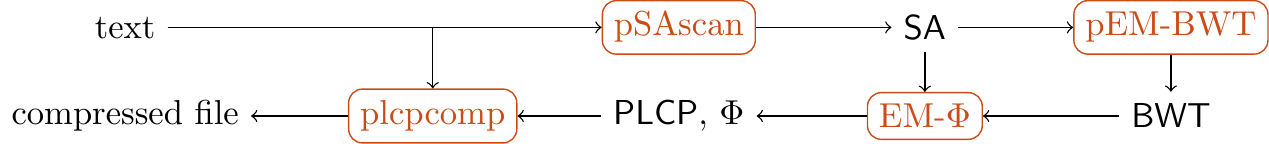}
    }%
    \caption{Construction schedule of the text data structures needed for \protect\iPlcpcomp{}.}
    \label{figPLCPschedule}
\end{figure}

\subparagraph*{Measurements and Results}

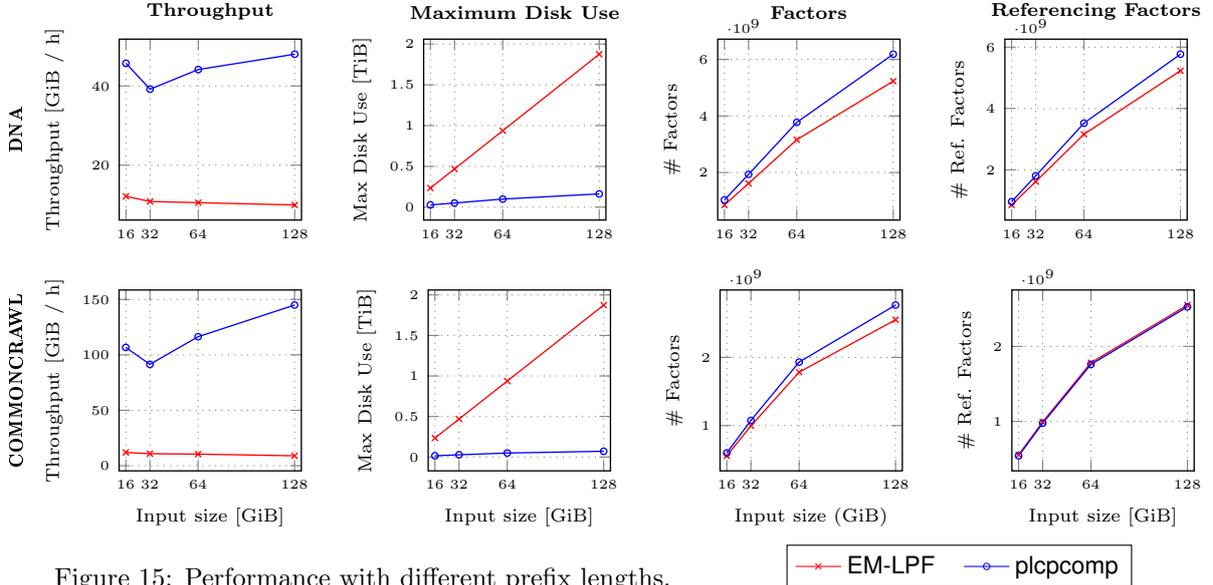
\begin{figure}
	\centering
	\begin{tikzpicture}
	\begin{axis}[plcpplots4,
	title={\bfseries Throughput},
	ylabel={\textsc{\bfseries DNA}\\[.75em]Throughput [GiB / h]\\[-.35em]},
	ylabel style={align=center},
	xtick={16,32,64,128},
	]
	
	\addplot coordinates { (16,12.1519) (32,10.8475) (64,10.5495) (128,9.97403) };
	\addlegendentry{chain=em-lpf};
	\addplot coordinates { (16,45.7143) (32,39.1837) (64,44.1379) (128,48.0) };
	\addlegendentry{chain=plcpcomp};
	\legend{}
	\end{axis}
	\end{tikzpicture}
	\hfill
	\begin{tikzpicture}
	\begin{axis}[plcpplots4,
	title={\bfseries Maximum Disk Use},
	ylabel={Max Disk Use [TiB]},
	xtick={16,32,64,128},
	]
	
	\addplot coordinates { (16,0.234376) (32,0.468751) (64,0.937501) (128,1.875) };
	\addlegendentry{chain=em-lpf};
	\addplot coordinates { (16,0.0271263) (32,0.0509014) (64,0.099268) (128,0.162624) };
	\addlegendentry{chain=plcpcomp};
	\legend{}
	\end{axis}
	\end{tikzpicture}
	\hfill
	\begin{tikzpicture}
	\begin{axis}[plcpplots4,
	title={\bfseries Factors},
	ylabel={\# Factors},
	xtick={16,32,64,128},
	]
	\addplot coordinates { (16,8.55084e+08) (32,1.61601e+09) (64,3.16044e+09) (128,5.2284e+09) };
	\addlegendentry{chain=em-lpf};
	\addplot coordinates { (16,1.03364e+09) (32,1.93675e+09) (64,3.77316e+09) (128,6.18054e+09) };
	\addlegendentry{chain=plcpcomp};
	\legend{}
	\end{axis}
	\end{tikzpicture}
	\hfill
	\begin{tikzpicture}
	\begin{axis}[plcpplots4,
	title={\bfseries Referencing Factors},
	ylabel={\# Ref. Factors},
	xtick={16,32,64,128},
	legend to name={leg:compare:dna},
	legend columns=2,
	legend style={font=\small},
	]
	
	\addplot coordinates { (16,8.55084e+08) (32,1.61601e+09) (64,3.16044e+09) (128,5.2284e+09) };
	\addlegendentry{chain=em-lpf};
	\addplot coordinates { (16,9.62039e+08) (32,1.80535e+09) (64,3.52074e+09) (128,5.76793e+09) };
	\addlegendentry{chain=plcpcomp};

	\legend{}
	\addlegendentry{\iEMLPF{}};
	\addlegendentry{\iPlcpcomp{}};
	\end{axis}
	\end{tikzpicture}\\[.75em]
	
	\begin{tikzpicture}
	\begin{axis}[plcpplots4,
	xlabel={Input size [GiB]},
	ylabel={\textsc{\bfseries COMMONCRAWL}\\[.75em]Throughput [GiB / h]},
	ylabel style={align=center},
	xtick={16,32,64,128},
	]
	
	\addplot coordinates { (16,12.0) (32,10.8475) (64,10.4918) (128,8.97196) };
	\addlegendentry{chain=em-lpf};
	\addplot coordinates { (16,106.667) (32,91.4286) (64,116.364) (128,144.906) };
	\addlegendentry{chain=plcpcomp};
	\legend{}
	\end{axis}
	\end{tikzpicture}
	\hfill
	\begin{tikzpicture}
	\begin{axis}[plcpplots4,
	xlabel={Input size [GiB]},
	ylabel={Max Disk Use [TiB]},
	xtick={16,32,64,128},
	]
	
	\addplot coordinates { (16,0.234376) (32,0.468751) (64,0.937501) (128,1.875) };
	\addlegendentry{chain=em-lpf};
	\addplot coordinates { (16,0.0153141) (32,0.0275536) (64,0.0496273) (128,0.0712452) };
	\addlegendentry{chain=plcpcomp};
	\legend{}
	\end{axis}
	\end{tikzpicture}
	\hfill
	\begin{tikzpicture}
	\begin{axis}[plcpplots4,
	xlabel={Input size (GiB)},
	ylabel={\# Factors},
	xtick={16,32,64,128},
	]
	
	\addplot coordinates { (16,5.56672e+08) (32,9.95691e+08) (64,1.78276e+09) (128,2.55056e+09) };
	\addlegendentry{chain=em-lpf};
	\addplot coordinates { (16,5.98973e+08) (32,1.07485e+09) (64,1.93019e+09) (128,2.7668e+09) };
	\addlegendentry{chain=plcpcomp};
	\legend{}
	\end{axis}
	\end{tikzpicture}
	\hfill
	\begin{tikzpicture}
	\begin{axis}[plcpplots4,
	xlabel={Input size [GiB]},
	ylabel={\# Ref. Factors},
	xtick={16,32,64,128},
	transpose legend=true,
	legend to name={leg:compare:cc},
	legend style={font=\small},
	]
	
	\addplot coordinates { (16,5.56658e+08) (32,9.95675e+08) (64,1.78274e+09) (128,2.55054e+09) };
	\addlegendentry{chain=em-lpf};
	\addplot coordinates { (16,5.43157e+08) (32,9.77253e+08) (64,1.76013e+09) (128,2.52689e+09) };
	\addlegendentry{chain=plcpcomp};
	
	\end{axis}
	\end{tikzpicture}
	
	\begin{minipage}{0.6\linewidth}
		\vspace{1em}
		\caption{Performance %
			with different prefix lengths.}
		\label{figResultsCCAndDNA}
	\end{minipage}
	\hfill
	\begin{minipage}{0.35\linewidth}
		\vspace{-1em}
		\ref{leg:compare:dna}
	\end{minipage}
\end{figure}

Our experiments measure the throughput, the maximum hard disk usage, and the number of referencing factors,
for \iEMLPF{} and \iPlcpcomp{} for $2^k \SI{}{\gibi\byte}$ prefixes ($4 \le k \le 7$) of our data sets \textsc{dna} and \textsc{commoncrawl}.
We collected the median of three iterations and present the results in \cref{figResultsCCAndDNA}.
The plots show that \iPlcpcomp{} is magnitudes faster on both data sets (cf.\ plots \bsq{Throughput}).
The reason for this could be that the disk accesses of \iEMLPF{} scale much worse than those of \iPlcpcomp{} (cf.\ plots \bsq{Maximum Disk Use}).
We point out that \iPlcpcomp{} is already faster than the step for computing \LCP{} from \PLCP{} and \SA{}.
Regarding the number of factors, \iPlcpcomp{} is on par with \iLZSS{} (rightmost plots), 
producing, relatively speaking, slightly more factors.

\subparagraph*{Decompression}
We ran our decompressor implementation on the \iPlcpcomp{} codings of our datasets.
Plots of the scaling experiments are shown in \cref{figResultsDecomp}.
As the decompression algorithm is superlinear, the throughput is decreasing with increasing text size. However, comparing the results for the 32GiB and 64GiB commoncrawl decompression, the throughput only decreases by $1\%$. The throughput between the 32GiB and 64 GiB DNA decompression differs by only $5\%$. The maximum external memory allocation rises linearly with increasing text size.

In \cref{figResultsTheta}, we measured the impact of the choice of $\threshold$ on the compressed output and the decompression algorithm of our datasets.
For larger values of $\threshold$, \iPlcpcomp{} creates less referencing factors, but the total number of factors increases (as we obtain much more literal factors). Having less referencing factors, the decompression needs less disk space.

Our decompression requires multiple sorting steps on the factor lists such as $\vecChld$ (cf.~\cref{sec:decompression}). 
The number of these steps depend on the maximum depth of (a tree in) the dependency graph induced by the factorization.
Therefore, it is not surprising that the decompressor is magnitudes slower than the comparatively simple compression algorithm.

Furthermore, and for the same reason, our decompression (expectedly) runs slower than the external memory Lempel-Ziv decoder of \citet{belazzougui16decoding}, which is why we skip a more detailed performance comparison here.

\section{Conclusions}
We presented \iPlcpcomp{}, the first external memory bidirectional compression algorithm, and showed its practicality by performing experiments on very large data sets, using only very limited RAM\@.
We also presented a decompression algorithm in external memory, which can decode the output of any bidirectional compression scheme (not only \iPlcpcomp{}).
Possible future steps include relating the number of factors of \iPlcpcomp{} to the minimal number of factors in a bi- or unidirectional compression scheme, evaluating the whole compression chain by also experimenting on \emph{codings} of the output of \iPlcpcomp{} (similar to \cite{dinklage17tudocomp}), and improving the performance of the decompression algorithm.

\begin{figure}
\centering
\begin{tikzpicture}
    \begin{axis}[plcpplots,
title={\bfseries Throughput},
xlabel={Original input size [GiB]},
ylabel={\textsc{\bfseries DNA}\\[.75em]Throughput [GiB / h]},
ylabel style={align=center},
xtick={16,32,64,128},
]
\addplot coordinates { (16,1.11047) (32,1.08505) (64,1.0317) };
\addlegendentry{algo=plcp-decomp};
\legend{}
\end{axis}
\end{tikzpicture}
\hfill
\begin{tikzpicture}
    \begin{axis}[plcpplots,
title={\bfseries Throughput},
xlabel={\# Factors},
ylabel={Throughput [\# Factors / h]},
ylabel style={align=center},
]
\addplot coordinates { (1.03364e+09,7.17389e+07) (1.93675e+09,6.56711e+07) (3.77316e+09,6.08247e+07) };
\addlegendentry{algo=plcp-decomp};
\legend{}
\end{axis}
\end{tikzpicture}
\hfill
\begin{tikzpicture}
\begin{axis}[plcpplots,
title={\bfseries Maximum Disk Use},
xlabel={Original input size [GiB]},
ylabel={Max Disk Use [TiB]},
xtick={16,32,64,128},
]
\addplot coordinates { (16,0.767418) (32,1.54679) (64,3.1301) };
\addlegendentry{algo=plcp-decomp};
\legend{}
\end{axis}
\end{tikzpicture}
\begin{tikzpicture}
    \begin{axis}[plcpplots,
xlabel={Original input size [GiB]},
ylabel={\textsc{\bfseries COMMONCRAWL}\\[.75em]Throughput [GiB / h]},
ylabel style={align=center},
xtick={16,32,64,128},
]
\addplot coordinates { (16,1.61616) (32,1.62299) (64,1.60602) };
\addlegendentry{algo=plcp-decomp};
\legend{}
\end{axis}
\end{tikzpicture}
\hfill
\begin{tikzpicture}
    \begin{axis}[plcpplots,
xlabel={\# Factors},
ylabel={Throughput [\# Factors / h]},
ylabel style={align=center},
]
\addplot coordinates { (5.98973e+08,6.05023e+07) (1.07485e+09,5.45146e+07) (1.93019e+09,4.84365e+07) };
\addlegendentry{algo=plcp-decomp};
\legend{}
\end{axis}
\end{tikzpicture}
\hfill
\begin{tikzpicture}
\begin{axis}[plcpplots,
xlabel={Original input size [GiB]},
ylabel={Max Disk Use [TiB]},
xtick={16,32,64,128},
]
\addplot coordinates { (16,0.603251) (32,1.21967) (64,2.3725) };
\addlegendentry{algo=plcp-decomp};
\legend{}
\end{axis}
\end{tikzpicture}
\caption{Performance of the decompression with different prefix lengths.}
\label{figResultsDecomp}
\end{figure}
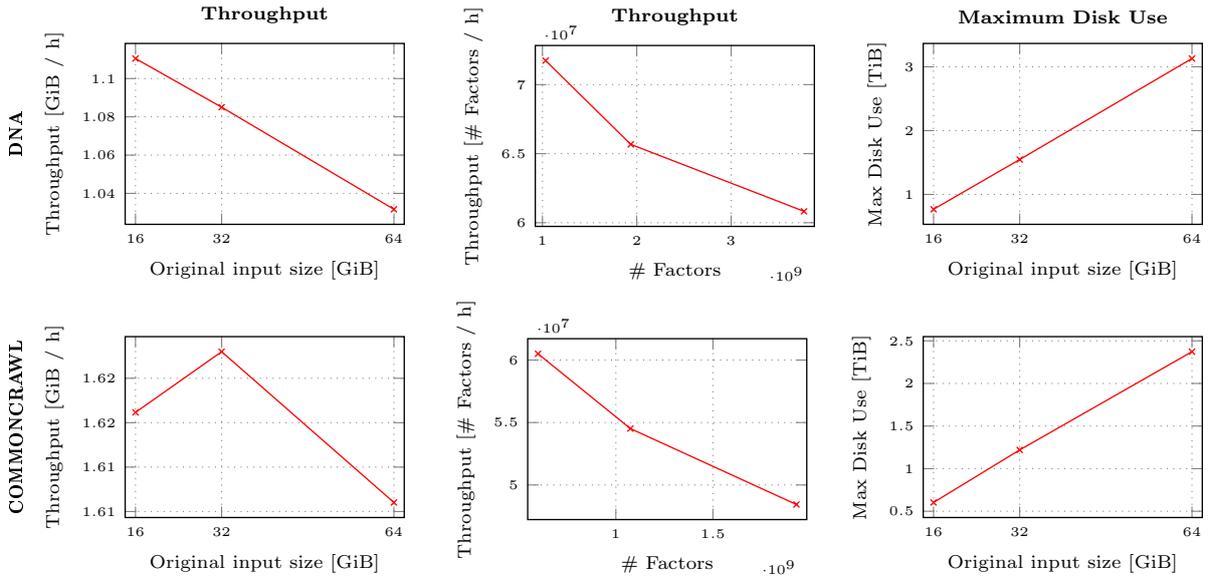

\begin{figure}
	\centering
	\begin{tikzpicture}
	\begin{axis}[plcpplots4,
	title={$\threshold\rightarrow$ Throughput},
	xlabel={$\threshold$},
	ylabel={\textsc{\bfseries DNA}\\[.75em] Throughput [GiB / h]},
	ylabel style={align=center},
	]
	\addplot coordinates { (17,60.0) (18,60.0) (19,60.0) (20,60.0) (21,60.0) (22,60.0) };
	\addlegendentry{size=16};
	\addplot coordinates { (17,48.0) (18,48.0) (19,48.0) (20,48.0) (21,48.0) (22,48.0) };
	\addlegendentry{size=32};
	\addplot coordinates { (17,44.6512) (18,44.6512) (19,44.6512) (20,44.1379) (21,44.1379) (22,43.6364) };
	\addlegendentry{size=64};
	\legend{}
	\end{axis}
	\end{tikzpicture}
	\begin{tikzpicture}
	\begin{axis}[plcpplots4,
	title={$\threshold\rightarrow$ Max Disk Use},
	xlabel={$\threshold$},
	ylabel={Max Disk Use [GiB]},
	]
	\addplot coordinates { (17,15.9531) (18,14.1113) (19,11.625) (20,9.3086) (21,7.62501) (22,6.46876) };
	\addlegendentry{size=16};
	\addplot coordinates { (17,31.1797) (18,28.6465) (19,24.7637) (20,20.4961) (21,17.0254) (22,14.5684) };
	\addlegendentry{size=32};
	\addplot coordinates { (17,60.8184) (18,56.9316) (19,50.9277) (20,43.375) (21,36.4336) (22,31.2266) };
	\addlegendentry{size=64};
	\legend{}
	\end{axis}
	\end{tikzpicture}
	\begin{tikzpicture}
	\begin{axis}[plcpplots4,
	title={$\threshold\rightarrow$ Factors},
	xlabel={$\threshold$},
	ylabel={\# Factors},
	]
	\addplot coordinates { (17,3.82768e+09) (18,4.84732e+09) (19,6.31292e+09) (20,7.7558e+09) (21,8.86394e+09) (22,9.6655e+09) };
	\addlegendentry{size=16};
	\addplot coordinates { (17,6.68603e+09) (18,8.09011e+09) (19,1.03754e+10) (20,1.30364e+10) (21,1.53204e+10) (22,1.70236e+10) };
	\addlegendentry{size=32};
	\addplot coordinates { (17,1.31503e+10) (18,1.5304e+10) (19,1.8839e+10) (20,2.35477e+10) (21,2.81163e+10) (22,3.17246e+10) };
	\addlegendentry{size=64};
	\legend{}
	\end{axis}
	\end{tikzpicture}
	\begin{tikzpicture}
	\begin{axis}[plcpplots4,
	title={$\threshold\rightarrow$ Ref. Factors},
	xlabel={$\threshold$},
	ylabel={\# Referencing Factors},
	]
	\addplot coordinates { (17,5.52495e+08) (18,4.88767e+08) (19,4.02556e+08) (20,3.22396e+08) (21,2.64072e+08) (22,2.23995e+08) };
	\addlegendentry{size=16};
	\addplot coordinates { (17,1.0799e+09) (18,9.92143e+08) (19,8.57712e+08) (20,7.09883e+08) (21,5.89668e+08) (22,5.0451e+08) };
	\addlegendentry{size=32};
	\addplot coordinates { (17,2.10649e+09) (18,1.97188e+09) (19,1.76394e+09) (20,1.50234e+09) (21,1.26189e+09) (22,1.08148e+09) };
	\addlegendentry{size=64};
	\legend{}
	\end{axis}
	\end{tikzpicture}
	
	\begin{tikzpicture}
	\begin{axis}[plcpplots4,
	title={$\threshold\rightarrow$ Throughput},
	xlabel={$\threshold$},
	ylabel={\textsc{\bfseries COMMONCRAWL}\\[.75em] Throughput [GiB / h]},
	ylabel style={align=center},
	]
	\addplot coordinates { (3,96.0) (4,96.0) (5,96.0) (6,96.0) (7,96.0) (8,96.0) };
	\addlegendentry{size=16};
	\addplot coordinates { (3,96.0) (4,101.053) (5,101.053) (6,96.0) (7,101.053) (8,96.0) };
	\addlegendentry{size=32};
	\addplot coordinates { (3,128.0) (4,128.0) (5,120.0) (6,128.0) (7,128.0) (8,128.0) };
	\addlegendentry{size=64};
	\legend{}
	\end{axis}
	\end{tikzpicture}
	\begin{tikzpicture}
	\begin{axis}[plcpplots4,
	title={$\threshold\rightarrow$ Max Disk Use},
	xlabel={$\threshold$},
	ylabel={Max Disk Use [GiB]},
	]
	\addplot coordinates { (3,14.7266) (4,13.8418) (5,13.0625) (6,12.3281) (7,11.5977) (8,10.8789) };
	\addlegendentry{size=16};
	\addplot coordinates { (3,26.5527) (4,25.0215) (5,23.666) (6,22.4102) (7,21.1836) (8,19.9649) };
	\addlegendentry{size=32};
	\addplot coordinates { (3,47.9238) (4,45.2617) (5,42.9121) (6,40.752) (7,38.666) (8,36.5977) };
	\addlegendentry{size=64};
	\legend{}
	\end{axis}
	\end{tikzpicture}
	\begin{tikzpicture}
	\begin{axis}[plcpplots4,
	title={$\threshold\rightarrow$ Factors},
	xlabel={$\threshold$},
	ylabel={\# Factors},
	]
	\addplot coordinates { (3,6.32111e+08) (4,6.93412e+08) (5,7.7444e+08) (6,8.76225e+08) (7,1.00245e+09) (8,1.15212e+09) };
	\addlegendentry{size=16};
	\addplot coordinates { (3,1.13244e+09) (4,1.23863e+09) (5,1.37929e+09) (6,1.55328e+09) (7,1.76587e+09) (8,2.01934e+09) };
	\addlegendentry{size=32};
	\addplot coordinates { (3,2.03047e+09) (4,2.21484e+09) (5,2.45907e+09) (6,2.75838e+09) (7,3.11944e+09) (8,3.54957e+09) };
	\addlegendentry{size=64};
	\legend{}
	\end{axis}
	\end{tikzpicture}
	\begin{tikzpicture}
	\begin{axis}[plcpplots4,
	legend to name={leg:theta},
	transpose legend=true,
	legend style={font=\small},
	legend columns=3,
	title={$\threshold\rightarrow$ Ref. Factors},
	xlabel={$\threshold$},
	ylabel={\# Referencing Factors},
	]
	\addplot coordinates { (3,5.10018e+08) (4,4.79368e+08) (5,4.52358e+08) (6,4.26912e+08) (7,4.01666e+08) (8,3.76722e+08) };
	\addlegendentry{size=16};
	\addplot coordinates { (3,9.19655e+08) (4,8.66561e+08) (5,8.19676e+08) (6,7.76179e+08) (7,7.3366e+08) (8,6.91415e+08) };
	\addlegendentry{size=32};
	\addplot coordinates { (3,1.65986e+09) (4,1.56767e+09) (5,1.48626e+09) (6,1.41143e+09) (7,1.33922e+09) (8,1.26753e+09) };
	\addlegendentry{size=64};
	\legend{}
	
	\addlegendentry{$16$ GiB input size };
	\addlegendentry{$32$ GiB input size };
	\addlegendentry{$64$ GiB input size };
	\end{axis}
	\end{tikzpicture}
	
	\ref{leg:theta}
	\caption{Evaluation of \protect\iPlcpcomp{} with different threshold values $\threshold$.}
	\label{figResultsTheta}
\end{figure}
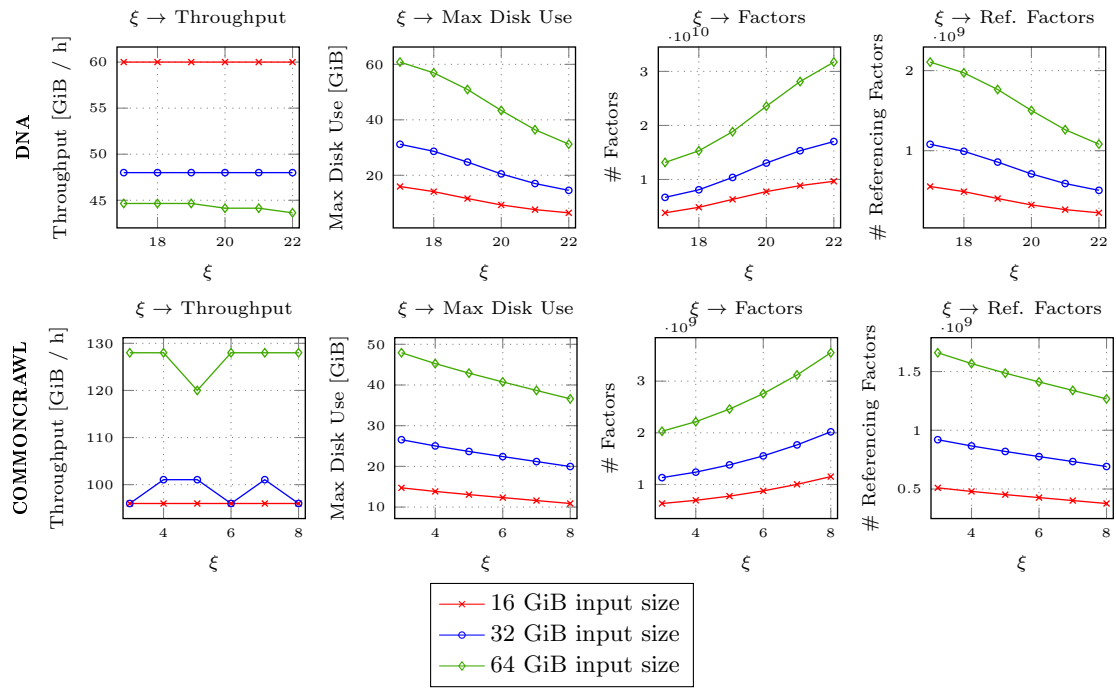

\bibliographystyle{plainnat}
\bibliography{literature}

\end{document}